\documentclass[11pt]{article}
\usepackage{amsmath,amsthm,amssymb}
\usepackage[english]{babel}
\usepackage{color}
\usepackage{graphicx}
\usepackage[top=3cm, bottom=3cm, left=2.3cm, right=2.3cm]{geometry}
\numberwithin{equation}{section}
\newtheorem{theorem}{Theorem}[section]
\newtheorem{lemma}[theorem]{Lemma}

\newtheorem{remark}[theorem]{Remark}

\def\bfR{{\mathbb R}}
\def\pp{\mbox{\bf P}}
\def\ee{\mbox{\bf E}}
\def\dd{\mbox{\rm d}}

\newcommand*{\fatdot}{ \raisebox{-1pt}{\makebox[.5ex]{\textbf{$\cdot$}}} }
\begin{document}
\title{On Trading American Put Options with Interactive Volatility}
\author{Sigurd Assing \& 
Yufan Zhao\footnote{Supported by EPSRC Ph.D.\ scholarship \#ASTAA1213YXZ.}\\
Department of Statistics, The University of Warwick,
Coventry CV4 7AL, UK\\
e-mail: s.assing@warwick.ac.uk}
\date{}
\maketitle
\begin{abstract}
We introduce a simple stochastic volatility model,
whose novelty consists in 
taking into account hitting times of the asset price,
and study the optimal stopping problem corresponding to a put option
whose time horizon (after the asset price hits a certain level)
is exponentially distributed.
We obtain explicit optimal stopping rules in various cases
one of which is interestingly complex because of
an unexpected disconnected continuation region.
Finally, we discuss in detail how these stopping rules could be used for trading
an American put
when the trader expects a market drop in the near future.
\end{abstract}
\noindent
{\large KEY WORDS}\hspace{0.4cm}
optimal stopping,
stochastic volatility model,
regime switching,
American option

\vspace{10pt}
\noindent
{\it Mathematics Subject Classification (2010)}:
Primary 60G40; Secondary 60J28
\section{Introduction and Results}\label{results}
This paper deals with an optimal stopping problem which is motivated
by option trading.

First, we introduce a simple short term model for the price of an asset
which is able to capture some aspects of the so-called leverage effect,
and, second, under such a model,
we predict value and exercise time 
of a perpetual American put written on this asset.

The leverage effect refers to the phenomenon that, typically, decreasing asset prices
are accompanied by rising volatility. We will not argue about whether the leverage effect
is a true phenomenon or not. We rather treat it as an observed phenomenon
which has been discussed in many papers since the mid 1970s when Black, \cite{black1976},
gave a well received macroeconomic explanation.
Since this effect has been observed, risk-seeking market participants
might want to take advantage of it.

However, other effects can superimpose a possible leverage effect.
For example, a decreasing stock price after a negative earning report
usually goes along with falling volatility as uncertainty decreases after
an announced event. Hence, the decision to bet on a combination of falling prices and 
rising volatility requires a careful analysis of relevant market conditions
which is left to the acting market participants.

The market participant we have in mind  is an option trader
who has made this decision and plans to go long on an American put.
The rationale behind going long on an American put when betting on a leverage effect is twofold;
falling prices and increasing volatility would both raise put prices.
But, if the trader wants to understand the risk of such a betting strategy before entering the trade, 
they should create a model for the price $(S_t,\,t\ge 0)$ 
of the asset underlying the American put which, 
first, is simple enough, second, is able to capture key features
of the trader's preferences for the future and, third, has enough parameters
to control the probabilities of different scenarios of future prices.

The model we propose can heuristically be described as follows:
\begin{itemize}\label{bullets}
\item
the price $S_t,\,t\ge 0$,
behaves like a geometric Brownian motion 
with volatility parameter $\sigma_0$ and trend $\mu_0$
until it hits a critical level $s_0\ll s$ where $s$ is the present price;
\item
when hitting the critical level,
the volatility parameter steps up to $\sigma_1$
and the trend of the stock changes to $\mu_1$;
\item
this `excited' state lasts for a period of length $T$ which is exponentially distributed 
with rate $\lambda$;
\item
finally the price is frozen at its value taken
when the exponential time $T$ has expired.
\end{itemize}

The \underline{new} feature of this volatility model 
is its dependence on hitting times of the price process
which has not been discussed in the literature, yet.
We call such a stochastic volatility {\it interactive volatility} to emphasise
this extra dependence.
\begin{remark}\rm\label{world}
\begin{itemize}\item[(i)]
The above model supposes $S_t=S_{t\wedge(\tau_{s_0}+T)}$, for $t\ge 0$, where
\begin{equation}\label{hittingDef}
\tau_{a}
\,\stackrel{\mbox{\tiny def}}{=}\,
\inf\{t\ge 0:S_t\le a\},
\end{equation}
for given price levels $a$.
The reason for freezing the price at $\tau_{s_0}+T$ is that
this time span is considered the time horizon of the trade:
the trader wants to be out once the market's volatility
has dropped back to normal.
Studying a \underline{perpetual} American put under this model 
easily reveals that the put
should be optimally exercised at the random time $\tau_{s_0}+T$, 
latest---see Remark \ref{discount}(ii).
\item[(ii)]
The notation $s_0\ll s$ is used to emphasise that the difference $s-s_0$
between the present price of the asset and the critical level should be chosen
big enough. The size of $s-s_0$ determines the strength of the market's drop
which causes the regime change from volatility $\sigma_0$ to $\sigma_1$
according to the leverage effect.
\item[(iii)]
A reasonable choice for $\sigma_0$ would be the implied volatility at present time
of the \underline{traded} American put the trader wants to long. Now recall that
\[
\log S_t\,=\,
\log s\,+(\mu_0-\frac{\sigma_0^2}{2})\,t\,
+\sigma_0\hspace{-.8cm}\overbrace{B_t}^{Brownian\;motion}
\]
is assumed to hold for $t\in[0,\tau_{s_0})$.
Hence, for fixed $\sigma_0$, the choice of $\mu_0$ determines
the distribution of the hitting time $\tau_{s_0}$. To meet the
preferences of the trader of a market-fall in the near future,
$\mu_0$ should be chosen sufficiently small to decrease 
the probability of large values of $\tau_{s_0}$.
But, the trader should also analyse the optimal stopping problem
for larger values of $\mu_0$, that is, they should analyse
their position under the assumption they are wrong
and the probability of a market-fall in the near future is rather small.
\item[(iv)]
For $t\in[\tau_{s_0},\tau_{s_0}+T)$, which is the final period of the trade,
the trader assumes
\[
\log S_t\,=\,
\log s_0\,+(\mu_1-\dfrac{\sigma_1^2}{2})(t-\tau_{s_0})
+\,\sigma_1(B_t-B_{\tau_{s_0}}).
\]
The choice of the parameters $\sigma_1,\mu_1$ reflects the trader's view
on the strength of the regime change triggered by the leverage effect,
and hence this choice is more or less subject to both the trader's
experience and their understanding of the market's history.
\item[(v)]
Using an exponential time $T$ for modelling the time span of the impact of the
leverage effect keeps the model simple enough. It is also assumed that $T$
is independent of what has happened before $\tau_{s_0}$.
The parameter $\lambda$ should be big enough to ensure that, on average,
the time span of the new volatility regime is of the order of days and not weeks.
If both $\tau_{s_0}$ and $T$ are on average rather short then the whole trade's
time horizon is likely to be less than the time to maturity of traded
American puts.
\item[(vi)]
Following the suggestions made in items (iii),(iv) above,
the trader's reasoning behind choosing $\mu_0,\mu_1$  has nothing to do 
with the market's rate of interest during the time span of their trade.
Thus, the model's underlying probability measure should be considered 
a guess of the real-world measure rather than a pricing measure.
Working out the optimal exercise time of the perpetual American put 
in the context of this model
gives the trader an indication of when to exit a trade they entered in accordance with 
their own preferences for the future.
The value of the put under the model is mainly used for finding the optimal exercise time,
and should NOT be confused with the price of a traded put.
\item[(vii)]
Our analysis can be used to motivate the choice of a traded American put 
with reasonable strike level and time to maturity
for the purpose of betting on a combination of falling prices 
and rising volatility---see Section \ref{numerical} for detailed examples.
\end{itemize}
\end{remark}

The proposed model has features of a Markov chain regime switching volatility model
as the excited state, when the volatility is $\sigma_1$, lasts for an exponential time.
But, at the end of this exponential time, 
instead of moving into a state which corresponds to another volatility level,
the Markov chain moves into an absorbing state. 
So, for the second and final period of the trade, the model can be regarded as
a degenerated Markov chain regime switching volatility model. 
We will comment on Markov chain regime switching volatility models
in Remark \ref{existingLit}(ii) below.

For the initial period of the trade, the model is different to a 
Markov chain regime switching volatility model as the system does not enter
the excited state following the move of a Markov chain. Instead, it enters
the excited state according to if the price of the asset has fallen to the
critical level $s_0$ or not, that is, according to 
how the price of the asset has behaved in the past.

To achieve Markovianity, 
we add a process $(Y_t,\,t\ge 0)$ for screening 
whether the price $S_t$ has already fallen to the critical level $s_0$ or not.
To fully describe the dynamics of $S_t$,
we also add a process $(\eta_t,\,t\ge 0)$ which 
is an absorbing Markov chain screening the length of the excitation.

Technically, we work with a strong Markov process,
$(S,Y,\eta)=(S_t,Y_t,\eta_t,t\ge 0)$,
on a family of probability spaces
$(\Omega,{\cal F},\pp_{\!\!s,y,i},(s,y,i)\in\bar{A})$,
where
\[
A\,\stackrel{\mbox{\tiny def}}{=}\,
\left[\rule{0pt}{12pt}\right.
(s_0,\infty)\times\{0\}\times\{1\}
\left.\rule{0pt}{12pt}\right]
\cup
\left[\rule{0pt}{12pt}\right.
(0,\infty)\times\{1\}\times\{0,1\}
\left.\rule{0pt}{12pt}\right]
\]
is considered a subset of the topological space 
$(0,\infty)\times\{0,1\}\times\{0,1\}$
equipped with the product topology.

The generator of this process is formally defined by
\[
\begin{array}{rcll}\label{big generator}
{L}f(s,0,1) 
&=&
\mu_0 s\partial_1f(s,0,1)+\frac{1}{2}\sigma_0^2 s^2 \partial_1^2 f(s,0,1),
&
\mbox{for $s\in(s_0,\infty),$}\\
\rule{0pt}{15pt}
{L}f(s,1,1) 
&=&
\mu_1 s\,\partial_1 f(s,1,1)+\frac{1}{2}\sigma_1^2 s^2\partial_1^2 f(s,1,1)
+\lambda \left[ f(s,1,0)-f(s,1,1) \rule{0pt}{10pt}\right],
&
\mbox{for $s\in(0,\infty),$}\\
\rule{0pt}{15pt}
{L}f(s,1,0) 
&=&
0,
&
\mbox{for $s\in(s,\infty),$}
\end{array}
\]
and considered an unbounded operator 
on the space $C_\infty(\bar{A})$ of continuous functions on $\bar{A}$
vanishing at infinity.
Its domain consists of all $f\in C_\infty(\bar{A})$ satisfying 
$f(s_0,1,1)-f(s_0,0,1)=0$ such that $Lf$,
when understood in the sense of Schwartz distributions on $A$,
gives a continuous function on $A$ vanishing at infinity
and with finite limits on $\partial A$.
\begin{remark}\rm\label{explain generator}
\begin{itemize}\item[(i)]
The condition $f(s_0,1,1)-f(s_0,0,1)=0$ is a discrete  Neumann boundary condition.
This boundary condition implies both
$\pp_{\!\!s_0,0,1}=\pp_{\!\!s_0,1,1}$ and
an interaction between the states
$(s_0,0,1)$ and $(s_0,1,1)$ leading to a jump of the process $Y_t$ 
when the price $S_t$ reaches $s_0$.
\item[(ii)]
As a consequence, for all $s>s_0$, under $\pp_{\!\!s,0,1}$, it holds that
$Y_t={\bf 1}_{[\tau_{s_0},\infty)}(t)$, $t\ge 0$,
and $\eta_t=1,\,t\le\tau_{s_0}$,
whereas,
for all $s>0$, under $\pp_{\!\!s,1,1}$, it holds that
$Y_t=1,\,t\ge 0$, and $(\eta_t,\,t\ge 0)$
is an independent two-states continuous-time Markov chain 
starting from one and absorbed at zero with rate $\lambda$.
\item[(iii)]
Combining (ii)
and 
${L}f(\cdot\,,1,0)\,\equiv\,0$ yields
\[
\pp_{\!\!s,y,i}
\left(\rule{0pt}{12pt}\right.\{
\mbox{
$S_t\,=\,S_{t\wedge\tau_{\eta,0}}$ for all  $t\ge 0$}
\}\left.\rule{0pt}{12pt}\right)
=\,1,
\quad\mbox{for all $(s,y,i)\in A$},
\]
where
\begin{equation}\label{T}
\tau_{\eta,0}\,\stackrel{\mbox{\tiny def}}{=}\,\inf\{t\ge 0:\eta_t=0\},
\end{equation}
that is, $\tau_{\eta,0}$ plays the role of what was called $\tau_{s_0}+T$
in Remark \ref{world}(i).
\item[(iv)]
Taking into account the other defining properties of the generator $L$,
the $S$-component of the process $(S,Y,\eta)$ started at $s>s_0$ has,
under $\pp_{\!\!s,0,1}$, the same law as the  price process
discussed  in items (iii) and (iv) of Remark \ref{world}.
Note that we could have worked with the process $(S,Y)$
killed at rate $\lambda$ after the jump of $Y$, instead.
But, as explained in Remark \ref{existingLit}(ii) below,
using an extra component like $\eta$ has the advantage that 
we can apply results on optimal stopping in the context of
Markov chain regime switching models.
\item[(v)]
Since $(S,Y,\eta)$ is strong Markov, the filtration $({\cal F}_t,\,t\ge 0)$
generated\footnote{Here,  {\it `filtration generated by $(S,Y,\eta)$\!'} refers to
the universal augmentation of the filtration
$(\sigma(S_u,Y_u,\eta_u:u\le t),\,t\ge 0)$
see Section 2.7.B of \cite{Karatzas} for a good account on universal filtrations.}
by $(S,Y,\eta)$
is right-continuous, and, by obvious reasons,
this filtration coincides with the smallest
right-continuous filtration which contains the universal augmentation of
$(\sigma(S_u:u\le t),\,t\ge 0)$.
\end{itemize}
\end{remark}

All in all, we have established a probability model 
predicting prices of an asset according to the features laid out 
in the four bullet points on page \pageref{bullets}.

Next, under this model, we will study value and optimal exercise time
of a perpetual American put contract written on this asset.
Using our probability model, such a put's value function takes the form
\begin{equation}\label{our problem}
V(s,y,i)
\,\stackrel{\mbox{\tiny def}}{=}\,
\sup_{\tau\ge 0}\,\ee_{s,y,i}[e^{-\alpha\tau}(K-S_{\tau})^+],
\quad\mbox{for $(s,y,i)\in A$},
\end{equation}
where the supremum is taken over all stopping times with respect to the filtration $({\cal F}_t,\,t\ge 0)$.
\begin{remark}\rm\label{discount}
\begin{itemize}\item[(i)]
The discount rate $\alpha$ refers to the rate of return of an investment 
the trader considers more or less riskless during the time interval of the trade. 
As explained in items (iii) and (v) of Remark \ref{world}, 
under the future preferences of the trader, this time interval
is supposed to be rather short on average, and hence choosing $\alpha$ to be constant 
is a good approximation. 
\item[(ii)]
Note that
\[
V(s,y,i)\,=\,
\sup_{\tau\le\tau_{\eta,0}}\,\ee_{s,y,i}[e^{-\alpha\tau}(K-S_{\tau})^+]
\]
because $S_t=S_{t\wedge\tau_{\eta,0}},\,t\ge 0$, implies
$e^{-\alpha\tau}(K-S_{\tau})^+\le e^{-\alpha\tau_{\eta,0}}(K-S_{\tau_{\eta,0}})^+$ 
on
$\tau\ge\tau_{\eta,0}$.
Hence, under our model, the perpetual put should be exercised at
$\tau_{\eta,0}=\tau_{s_0}+T$, latest.
\end{itemize}
\end{remark}

First recall the results for perpetual American put options 
as obtained in \cite{Mckean} in the context of geometric Brownian motion.
\begin{theorem}\label{18} 
Given on a family of probability spaces
$(\tilde{\Omega},\tilde{\mathcal{F}},\tilde{\pp}_s,s>0)$,
let $(\tilde{S}_t,\,t\ge 0)$ be the Feller process
whose generator is the closure of
\[
\tilde{L}f=\mu_0s f'+\dfrac{1}{2}\sigma_0^2s^2f'',
\quad
f\in C_0^2((0,\infty)),
\]
where the subscript zero, when added to $C^2$, means compact support.
Then, the value function
\[
\tilde{V}(s)
\,\stackrel{\mbox{\tiny def}}{=}\,
\sup\left\{\tilde{\ee}_s[e^{-\alpha\tilde{\tau}}(K-\tilde{S}_{\tilde{\tau}})^+]:
\mbox{$\tilde{\tau}$ stopping time with respect to
$(\tilde{S}_t,\,t\ge 0)$}\right\}
\]
is given by
\begin{equation}
\tilde{V}(s) = 
\begin{cases}
K-s &:\quad s\in(0,b_0]\\
(K-b_0)\bigg(\dfrac{s}{b_0}\bigg)^{\gamma^-} &:\quad s\in(b_0,\infty)
\end{cases}\label{17}
\end{equation}
where $b_0=-\gamma^- K/(1-\gamma^-)$, and
$\gamma^-$ stands for the negative root of the quadratic equation
\begin{equation}\label{30}
\frac{1}{2}\sigma_0^2\gamma^2+(\mu_0-\frac{1}{2}\sigma_0^2)\gamma-\alpha=0.
\end{equation}
\end{theorem}
\begin{remark}\rm\label{meaning b0}
\begin{itemize}\item[(i)]
The above value function $\tilde{V}$ satisfies 
\begin{equation*}
0\,=\,\mu_0 s\tilde{V}'(s)+
\frac{1}{2}\sigma_0^2 s^2\tilde{V}''(s)
-\alpha\tilde{V}(s), \quad\text{ for } s>b_0,
\end{equation*}
subject to
\[
\tilde{V}(b_0)=K-b_0,
\quad
\tilde{V}'(b_0)=-1,
\quad
\lim_{s\to\infty}\tilde{V}(s)=0.
\]
\item[(ii)]
If $\mu_0-\frac{1}{2}\sigma_0^2>0$, 
then there is no (finite) optimal stopping time 
at which the value function $\tilde{V}$ can be attained. 
But $\tilde{\tau}_{b_0}$ is a Markov time at which 
$\tilde{V}$ given in Theorem \ref{18}
is attained when setting
$\tilde{\ee}_s[e^{-\alpha\tilde{\tau}_{b_0}}(K-\tilde{S}_{\tilde{\tau}_{b_0}})^+]$
to be zero
on $\{\tilde{\tau}_{b_0}=\infty\}$.
In all further cases below, 
attaining a value function at
a possibly infinite Markov time will be understood as above,
since all considered value functions vanish at infinity.
\end{itemize}
\end{remark}

The next theorem presents the main result of this paper.
\begin{theorem}\label{main}
Recall (\ref{hittingDef}),\,(\ref{T}), 
and Remark \ref{world}(ii) for the purpose of $s_0$,
and Remark \ref{meaning b0}(ii) for the meaning of $b_0$.
Let $\gamma^+$ ($\gamma^-$) denote the positive (negative) root of 
equation (\ref{30}).
The following cases completely describe the value function
given by (\ref{our problem}).
\begin{itemize}
\item[(i)]
In the trivial case,
\[
V(s,1,0)=(K-s)^+,\quad\mbox{for all $s>0$},
\]
and the optimal stopping time is $0$.
\item[(ii)] Let $\beta^+$ ($\beta^-$) denote the positive (negative) 
root of the quadratic equation
\begin{equation}
\frac{1}{2}\sigma_1^2\beta^2+(\mu_1-\frac{1}{2}\sigma_1^2)\beta-(\alpha+\lambda)=0.\label{100}
\end{equation}
Then, there exists a smooth function $h:(0,\infty)\to\bfR$ such that
\begin{equation*}
V(s,1,1)=\begin{cases}
K-s &:\quad s\in(0,b_1]\\
c_1s^{\beta^+}+c_2s^{\beta^-}
+h(s) &:\quad s\in(b_1,K]\\
d_2s^{\beta^-} &:\quad s\in(K,\infty)
\end{cases}
\end{equation*}
where the coefficients $c_1$, $c_2,\,d_2$ 
and the stopping level $b_1$ are obtained by solving 
the equations (\ref{99}) on page \pageref{99}. 
The finite optimal stopping time is $\tau_{b_1}\wedge\tau_{\eta,0}$.
\item[(iii)]If one of the conditions
\begin{itemize}
\item[(a)]
$b_0> s_0$ and $V(s_0,1,1)\geq (K-b_0)({s_0}/{b_0})^{\gamma^-}$,
\item[(b)]
$b_0\leq s_0$ and $V(s_0,1,1)>(K-s_0)^+$,
\item[(c)]
$b_0\leq s_0<K$ and $V(s_0,1,1)=(K-s_0)$,
\end{itemize}
is satisfied, then
\begin{equation*}
V(s,0,1)=V(s_0,1,1)\bigg(\frac{s}{s_0}\bigg)^{\gamma^-}, 
\quad\text{for } s>s_0.
\end{equation*}
This value function is attained at the possibly infinite Markov time
$\tau_{b_1}\wedge\tau_{\eta,0}$, if (a),(b),
and $\tau_{s_0}$, if (c).
\item[(iv)] If $b_0>s_0$ 
and $K-s_0\le V(s_0,1,1)<(K-b_0)({s_0}/{b_0})^{\gamma^-}$, then
\begin{equation*}
V(s,0,1) = \begin{cases}
e_1^* s^{\gamma^+}+\;e_2^* s^{\gamma^-} &:\quad s\in  (s_0,b_*)\\
\rule{0pt}{20pt}
K-s &:\quad s\in [b_*, b_0]\cap(s_0,\infty)\\
(K-b_0)\bigg(\dfrac{s}{b_0}\bigg)^{\gamma^-} &:\quad s\in(b_0,\infty) \\
\end{cases}
\end{equation*}
where\footnote{Here $(s_0,b_*)=\emptyset$ by definition when $b_*=s_0$.}
$b_*=s_0$, if $K-s_0=V(s_0,1,1)$,
and $e_1^*,\,e_2^*,\,b_*$ 
taken from the proof of Lemma \ref{choice b*} on page \pageref{choice b*}, 
otherwise.
The value function is attained at the possibly infinite Markov time
$\tau_{[b_*,b_0],0}\wedge\tau_{b_1}\wedge\tau_{\eta,0}$, where
\[
\tau_{[a,b],0}\stackrel{\mbox{\tiny def}}{=} \inf\{t\ge 0:S_t\in[a,b],Y_t=0\},
\]
for levels $0<a<b$.
\end{itemize}
\end{theorem}
\begin{remark}\rm\label{existingLit}
\begin{itemize}\item[(i)]
The contribution of this paper consists  in the two non-trivial cases (iii) and (iv),
but also in the critical review of the known case (ii)---see 
Remark \ref{existingLit}(iii+iv) below.
Case (iv) is mathematically most interesting because 
it leads to a disconnected continuation region.
As explained in Remark \ref{world}(vi),
only the optimal stopping times are of true relevance for the trader.
Note that there are no further sub-cases than those mentioned under (iii) and (iv).
The restriction to $s_0<K$ in case (iii)(c) follows from
the fact that $V(s_0,1,1)>0$ if $s_0\ge K$.
We will discuss numerical examples for all sub-cases in Section \ref{numerical}.
\item[(ii)]
The calculations in
both cases (iii) and (iv) rely upon
the solution of the optimal stopping problem under the measures
$\pp_{\!\!s,1,1},\,s>0$, obtained under case (ii).
Under these measures, the price process $(S_t,\,t\ge 0)$
satisfies the stochastic differential equation
\[
\dd S_t\,=\,\mu_1\eta_t S_t\,\dd t+\sigma_1\eta_t S_t\,\dd B_t
\]
which is a consequence of what was discussed in 
Remark \ref{world}(iv) and (ii)+(iii) of Remark \ref{explain generator}.
Note that the above stochastic differential equation describes a 
Markov modulated geometric Brownian motion which
has extensively been used in the literature as a simple 
Markov chain regime switching volatility model.
However, there are just two important papers on pricing perpetual American puts
in the context of such a regime switching model:
one by Guo/Zhang, \cite{Guo}, treating the case of two-states Markov chains,
and another by Jobert/Rogers, \cite{Rogers}, 
treating the general case of Markov chains with finitely many states.
Buffington/Elliott's, \cite{BuffEll}, work on finite time horizon American puts
can be considered an extension of the ideas behind \cite{Guo}.
The explicit form of the value function given in case (ii)
is a special case of Guo/Zhang's result when the Markov chain is degenerated.
The function $h:(0,\infty)\to\bfR$ 
can be any solution to the first of the two differential equations
above (\ref{5}) on page \pageref{5}, for example,
\[
h(s)\,=\,
-\dfrac{\lambda s}{\alpha+\lambda-\mu_1}
+\dfrac{\lambda K}{\alpha+\lambda}
\quad\mbox{and}\quad
h(s)\,=\,
\dfrac{\lambda s\log s}{\alpha+\lambda+\sigma_1^2/2}
+\dfrac{\lambda K}{\alpha+\lambda}
\]
for $\alpha+\lambda\not=\mu_1$ and $\alpha+\lambda=\mu_1$, respectively.
\item[(iii)]
To prove that the explicit expression given for $V(\cdot,1,1)$
in case (ii) is indeed the value function requires a verification argument.
Denoting this explicit expression by $V^*(\cdot,1,1)$,
the standard method of verification would be
to verify the properties (v1),(v2),(v3) as given at the beginning
of Section \ref{24}, but with respect to the measures $\pp_{\!\!s,1,1},\,s>0$.
Following \cite{Guo}, one realises that the authors do not verify but assume (v3)
in their Theorem 3.1. Furthermore, their proof of (v2) is incomplete
as they only justify $(\alpha-L)V^*(\cdot,1,1)\le 0$ inside of the continuation region.
Outside of the continuation region, that is, when $V^*(\cdot,1,1)$ coincides with 
$(K-\cdot)^+$, the validity of $(\alpha-L)V^*(\cdot,1,1)\le 0$ would depend on 
the value of $b_1$, and this issue was not addressed in \cite{Guo}.
Finally they also assume $\mu_1\ge 0$ which is an assumption
we might want to avoid as explained in Remark \ref{world}.
All in all, the verification of the value function given in \cite{Guo} is incomplete
and their assumptions are too restrictive for our purpose,
and that's why we decided to address this verification in the Appendix.
\item[(iv)]
We also checked the verification arguments given in the proof of
\cite[Prop.2]{Rogers}. Here, the authors in particular demonstrate
$(\alpha-L)V^*(\cdot,1,1)\le 0$ outside of the continuation region but their
proof requires $(e^{-\mu_1 t}S_t,\,t\ge 0)$ to be a supermartingale
which restricts the choice of $\mu_1$ to $\mu_1\le r$.
We cannot follow their proof of (v3) as it looks to us as if they applied
Doob's optional sampling theorem using an unbounded stopping time.
Furthermore, we do not see why, without using further arguments,
the non-linear equations mentioned in \cite[Problem 2]{Rogers}
should have unique solutions (see Remark {\bf A.1} in the Appendix
for further details).
\end{itemize}
\end{remark}
\section{Proofs}
Both cases (i) and (ii) of Theorem \ref{main} assume $y=1$ which leads to
a special version of the results obtained in \cite{Guo}
where American puts were priced in the context of a 
two-states Markov chain volatility model. 
In our case, the $Q$-matrix of the corresponding Markov chain is degenerated.

Therefore, in the next section, we only sketch the proof of case (ii).
But we give enough details to put the notation used in Theorem \ref{main}(ii)
into context.
However, recall Remark \ref{existingLit}(iii) where we explained that the
verification of the value function in \cite{Guo} was incomplete.
The corresponding details can be found in the Appendix.
\subsection{Proof of Theorem \ref{main} (i) and (ii)}\label{guo case}
For all $s>0$, and any stopping time $\tau$, 
\begin{equation}
\ee_{s,1,i}[e^{-\alpha\tau}(K-S_{\tau})^+]\leq\ee_{s,1,i}[e^{-\alpha(\tau\wedge\tau_{\eta,0})}(K-S_{\tau\wedge \tau_{\eta,0}})^+], \quad\text{for } (s,i)\in(0,\infty)\times\{0,1\},\nonumber
\end{equation}
since the process $(S_t, t \geq 0)$ is stopped at $\tau_{\eta,0}$. 
Hence $V(s,1,0)=(K-s)^+$ with optimal stopping time 0 proving (i).

For showing (ii), recall that $Y_t=1$, for all $t\ge 0$, a.s., 
when starting the dynamics from any $s>0,\,y=i=1$.
Now,
assume that the stopping region takes the form
\[
(0,b_1]\times \{1\}\times\{1\} \cup (0,\infty)\times\{1\}\times\{0\}
\]
when starting from $s>0,\,y=i=1$, where $b_1$ is an unknown stopping level.
Then, by \cite[Theorem 2.4]{Peskir} for example, 
if $V$ is lower semi-continuous,
the value function $V(s,1,1)$ would be attained at
\begin{equation}\label{46}
\tau^* = \tau_{b_1}\wedge\tau_{\eta,0}\,,
\end{equation}
and
$(e^{-\alpha(\tau^*\wedge t)}V(S_{t\wedge\tau^*},Y_{t\wedge\tau^*},\eta_{t\wedge\tau^*}),
\,t\ge 0)$ would be a $\pp_{\!\!s,1,1}$-martingale.

We want to use this martingale property to derive equations 
for both $V$ and $b_1$. Assume for now that $V$ has even more regularity
and a generalised It\^o's formula (see Remark \ref{our ito} below) can be applied to obtain
\[
e^{-\alpha(t\wedge\tau^*)}V(S_{t\wedge\tau^*},Y_{t\wedge\tau^*},\eta_{t\wedge\tau^*})=V(S_0,Y_0,\eta_{0})+\int^{t\wedge\tau^*}_0 e^{-\alpha u}(L-\alpha I)V(S_u,Y_u,\eta_{u})\text{d}u + M_{t\wedge\tau^*},\label{43}
\]
for all $t\ge 0$, a.s., 
where $M$ stands for a local martingale and $I$ denotes the identity operator.

Of course, for the above left-hand side to be a martingale, 
the integral on the right-hand side must vanish.
Using both the specific form of $L$ as given on page \pageref{big generator}
and (i) proven above, a sufficient condition for this integral to vanish is
\[
\begin{cases}
0\,=\,
\mu_1 s\partial_1V(s,1,1)
+\frac{1}{2}\sigma_1^2 s^2\partial_{11}V(s,1,1)
+\lambda(K-s)
-(\alpha+ \lambda )V(s,1,1) &\text{for } s\in (b_1,K) \\
0\,=\,\rule{0pt}{20pt}
\mu_1 s\partial_1V(s,1,1)
+\frac{1}{2}\sigma_1^2 s^2\partial_{11}V(s,1,1)
-(\alpha+\lambda )V(s,1,1) &\text{for } s\in (K,\infty)
\end{cases}
\]
depending on the unknown $b_1$ subject to the boundary and pasting conditions
\begin{align}
\lim_{s\to\infty}V(s,1,1)&=0,\nonumber\\
\rule{0pt}{0pt}
V(K{-},1,1)&=V(K{+},1,1),\nonumber\\
\rule{0pt}{15pt}
\partial_1V(K{-},1,1)&=\partial_1V(K{+},1,1),\label{5}\\
\rule{0pt}{15pt}
K-b_1&=V(b_1+,1,1),\nonumber\\
\rule{0pt}{15pt}
-1&=\partial_1 V(b_1+,1,1),\nonumber
\end{align}
where $-$ and $+$ indicate taking left and right limits at
the corresponding argument, respectively.

The well-known solution of the above equation has the form
\begin{equation}\label{classical solution}
V(s,1,1)=\begin{cases}
c_1s^{\beta^+}+\;c_2s^{\beta^-}
+\;h(s) & \text{for } s\in(b_1,K)\\
d_1s^{\beta^+}+\;d_2s^{\beta^-} & \text{for } s\in(K,\infty)
\end{cases}
\end{equation}
with unknown coefficients $c_1$, $c_2$, $d_1$, $d_2$,
and $\beta^+$ ($\beta^-$) being the positive (negative) root of equation (\ref{100}).
For the choice of the function $h$ we refer to Remark \ref{existingLit}(ii).

Certainly, the first of the conditions under (\ref{5}) implies $d_1=0$,
and the other four conditions yield
\begin{equation}\label{99}
\begin{array}{rcl}
c_1K^{\beta^+}+\;c_2K^{\beta^-}
+\;h(K)
&=&
d_2K^{\beta^-},\\
\rule{0pt}{20pt}
c_1\beta^+K^{\beta^+}+\;c_2\beta^-K^{\beta^-}+\;K h'(K)
&=&
d_2\beta^-K^{\beta^{-}},\\
\rule{0pt}{20pt}
K-b_1
&=&
c_1b_1^{\beta^+}+\;c_2b_1^{\beta^-}+\;h(b_1),\\
\rule{0pt}{20pt}
-b_1
&=&
c_1\beta^{+}b_{1}^{\beta^+}+\;c_2\beta^{-}b_1^{\beta^-}
+\;b_1 h'(b_1).
\end{array}
\end{equation}

Note that the coefficients $c_1$, $c_2$, $d_2$ 
linearly depend on $b_1^{\beta^\pm}$
so that the problem comes down to 
solving numerically for $b_1$.
For verification we refer to the Appendix.

This concludes the discussion of $V$ for $y=1$. 
We now turn to the cases (iii) and (iv) of Theorem \ref{main}
dealing with the case $y=0$. 
\subsection{Proof of Theorem \ref{main} (iii)}
Recall the setup of Theorem \ref{18}, but also introduce
\[
\tilde{V}_0(s)
\,\stackrel{\mbox{\tiny def}}{=}\,
\sup\left\{\tilde{\ee}_s[e^{-\alpha\tilde{\tau}}
(K-\tilde{S}_{\tilde{\tau}\wedge\tilde{\tau}_{s_0}})^+]:
\mbox{$\tilde{\tau}$ stopping time with respect to
$(\tilde{S}_t,\,t\ge 0)$}\right\},
\]
for all $s\ge s_0$.
\begin{lemma}\label{23}
If $b_0\leq s_0<K$, then
\begin{equation}
\tilde{V}_0(s) =
(K-s_0)\bigg(\dfrac{s}{s_0}\bigg)^{\gamma^-},
\quad \text{for } s \geq s_0,
\label{19}\end{equation}
and the (possibly infinite) Markov time
$\tilde{\tau}_{s_0}$ is the optimal time.
\end{lemma}
\begin{proof}
Assume $b_0\le s_0<K$ and fix $s\ge s_0$.
By `guess and verify',
it suffices to check that the right-hand side of (\ref{19}) satisfies
\begin{itemize}
\item[(v1)]
$(K-s_0)(s/s_0)^{\gamma^-}=\,\tilde{\ee}_s[e^{-\alpha\tilde{\tau}_{s_0}}
(K-\tilde{S}_{\tilde{\tau}_{s_0}})^+
\,{\bf 1}_{\{\tilde{\tau}_{s_0}<\infty\}}]$,
\item[(v2)]
the process
$(e^{-\alpha t}(K-s_0)(\tilde{S}_{t\wedge\tilde{\tau}_{s_0}}/s_0)^{\gamma^-},\,t\ge 0)$
is a $\tilde{\pp}_{\!\!s}$-supermartingale,
\item[(v3)]
$(K-s_0)(s/s_0)^{\gamma^-}\geq (K-s)^+$.
\end{itemize}

For (v1), by It\^{o}'s formula,
\begin{align*}
&\tilde{\ee}_s[e^{-\alpha\tilde{\tau}_{s_0}}
(K-\tilde{S}_{\tilde{\tau}_{s_0}})^+
\,{\bf 1}_{\{\tilde{\tau}_{s_0}<\infty\}}]\\
=&\lim\limits_{t\rightarrow\infty}
\tilde{\ee}_s[\;e^{-\alpha(\tilde{\tau}_{s_0}\wedge t)}(K-s_0)
\left(\frac{\tilde{S}_{\tilde{\tau}_{s_0}\wedge t}}{s_0}\right)^{\gamma^-}\!]\\
=&\lim\limits_{t\rightarrow\infty}\tilde{\ee}_s\!\left[
(K-s_0)\left(\frac{s}{s_0}\right)^{\gamma^-}\!\!
+\;(K-s_0)
\int^{\tilde\tau_{s_0}\wedge t}_0 e^{-\alpha u}\,
[\;(\tilde{L}-\alpha I)\left(\frac{\fatdot}{s_0}\right)^{\gamma^-}\!]
(\tilde{S}_{u})\,du+M_{t\wedge\tilde\tau_{s_0}}
\right]\\
=&\,(K-s_0)\left(\frac{s}{s_0}\right)^{\gamma^-},
\end{align*}
because $(M_t,t\geq 0)$ is a $\tilde{\pp}_{\!\!s}$-martingale and 
the expression inside the integral vanishes.

For (v2), by Markov Property, it suffices to prove that
$\tilde{\ee}_s[e^{-\alpha t}(\tilde{S}_{t\wedge\tilde{\tau}_{s_0}}/s_0)^{\gamma^-}]
\leq (s/s_0)^{\gamma^-}$, for all $t\ge 0$, ignoring the constant $K-s_0$. 
But, for fixed $t\ge 0$,
\[
\tilde{\ee}_s[\;e^{-\alpha t}
\left(\frac{\tilde{S}_{t\wedge\tilde{\tau}_{s_0}}}{s_0}\right)^{\gamma^-}\!]
\leq 
\tilde{\ee}_s[\;e^{-\alpha(t\wedge\tilde{\tau}_{s_0})}
\left(\frac{\tilde{S}_{t\wedge\tilde{\tau}_{s_0}}}{s_0}\right)^{\gamma^-}\!]
=\left(\frac{s}{s_0}\right)^{\gamma^-}\!,
\]
where the last equality was already verified when proving (v1) above.

For (v3), note that
\[
\tilde{V}'(b_0)=(K-b_0)\,\dfrac{\gamma^-}{b_0}=-1
\]
as mentioned in Remark \ref{meaning b0}(i).
Based on two arguments, we can now
deduce that the derivative of $(K-s_0)(\textbf{$\cdot$}/s_0)^{\gamma^-}$ 
is bigger than $-1$ on $(s_0,K)$.
First, the derivative of $(K-s_0)(\textbf{$\cdot$}/s_0)^{\gamma^-}$
is bounded below by $-1$ at $s_0$, because $b_0\le s_0$ implies
\begin{equation*}
\dfrac{\tilde{V}'(b_0)}{(K-s_0)\gamma^-/s_0}=\dfrac{(K-b_0)s_0}{(K-s_0)b_0}\ge 1.
\end{equation*}
Second, $(K-s_0)(\textbf{$\cdot$}/s_0)^{\gamma^-}$ is convex on $(s_0,K)$.

But, if the derivative of $(K-s_0)(\textbf{$\cdot$}/s_0)^{\gamma^-}$ is bigger than $-1$
on $(s_0,K)$ and $(K-s_0)(\textbf{$\cdot$}/s_0)^{\gamma^-}$ 
touches $(K-\textbf{$\cdot$})$ at $s_0$,
then
$(K-s_0)(s/s_0)^{\gamma^-}>(K-s)^+$, for all $s\in (s_0,K)$.
Finally,
$(K-s_0)(s/s_0)^{\gamma^-}>(K-s)^+=0$, for $s\in[K,\infty)$, is obvious. 
\end{proof}
\begin{remark}\label{our ito}\rm
In what follows, we are going to use an easy application of 
Meyer's, \cite{meyer}, generalised It\^{o}'s formula which goes as
follows:
if $\phi:(0,\infty)\to\bfR$ is a function which is
twice continuously differentiable, except at finitely many points
$\{a_1,\dots,a_n\}$, such that
\[
\phi'(a_k\pm)\,\stackrel{\mbox{\tiny def}}{=}\,
\lim_{x\to a_k\pm}\phi'(x)
\quad\mbox{and}\quad
\phi''(a_k\pm)\,\stackrel{\mbox{\tiny def}}{=}\,
\lim_{x\to a_k\pm}\phi''(x)
\]
exist and are finite, $k=1,\dots,n$, 
then\footnote{Note that $(S_t,\,t\ge 0)$ only takes positive values.}
\[
\phi(S_t)\,=\,\phi(S_0)\,+
\int_0^t\phi'(S_u)\,\dd S_u\,+
\frac{1}{2}\int_0^t\phi''(S_u)\,\dd\langle S\rangle_u\,+
\sum\limits_{k=1}^n\frac{1}{2}\,L_t(a_k)\,
[\phi'(a_k+)-\phi'(a_k-)],
\]
for all $t\ge 0$, a.s.,
where $L_t$ stands for the local time of the 
continuous semimartingale $(S_t,\,t\ge 0)$.
Note that both integrands in the above formula are well-defined
for Lebesgue almost every $u$, almost surely, which uniquely determines
the integrals. 

If $\phi'$ is continuous then
the local time terms would even vanish and the above formula would 
look like the classical It\^o's formula.
\end{remark}

We now return to our problem of finding $V(s,0,1)$ for $s>s_0$. 
Recall that we already know $V(s,1,i)$, for all $s>0$, and $i=0,1$. 
\subsubsection{Proof of Theorem \ref{main} (iii)(a) and (iii)(b)}\label{24}
Suppose that either condition (a) or condition (b) of Theorem \ref{main}(iii) is satisfied.
Recall $\tau^*=\tau_{b_1}\wedge\tau_{\eta,0}$, and introduce:
\[
V^*(s,y,i)\,\stackrel{\mbox{\tiny def}}{=}\,
\left\{\begin{array}{lcl}
V(s_0,1,1)\left(\dfrac{s}{s_0}\right)^{\gamma^-}
&:&s>s_0,y=0,i=1\\
V(s,1,i)&:&s>0,y=1,i\in\{0,1\}
\end{array}\right.
\]

Verifying, for any $s>s_0$,
\begin{itemize}
\item[(v1)]
$V^*(s,0,1)=
\ee_{s,0,1}[e^{-\alpha\tau^*}(K-S_{\tau^*})^+\,{\bf 1}_{\{\tau^*<\infty\}}]$,
\item[(v2)]
the process $(e^{-\alpha t}\,V^*(S_t,Y_t,\eta_t),\,t\ge 0)$ is a 
$\pp_{\!\!s,0,1}$\,-\,supermartingale,
\item[(v3)]
$V^*(s,0,1)\ge(K-s)^+$,
\end{itemize}
would imply the conclusion of Theorem \ref{main} in both cases (a) and (b).

We are going to verify (v1),(v2),(v3).
First observe that
$V(s_0,1,1)>(K-s_0)^+$ holds in both cases (a) and (b).
To see this in the non-trivial case (a),
note that $(K-b_0)(\cdot/b_0)^{\gamma^-}$ is strictly convex 
on $[s_0,b_0]$ and touches $(K-s)^+$ at $s=b_0<K$, 
so $V(s_0,1,1)\geq(K-b_0)({s_0}/{b_0})^{\gamma^-}>K-s_0$. 
 
As a consequence,
$(s_0,1,1)$ is in the continuation region with respect to the optimal
stopping problem (\ref{our problem}) on page \pageref{our problem}.
Since $b_1$ is the lower boundary of the continuation region
when $y=i=1$, we can deduce that $b_1<s_0$, and hence $\tau_{s_0}<\tau^*$, 
for $s>s_0$. Therefore,
\begin{align*}
&\ee_{s,0,1}[e^{-\alpha\tau^*}(K-S_{\tau^*})^+\,
\mathbf{1}_{\{\tau_{s_0}<\infty\}}]\\
=&\ee_{s,0,1}
\left[
\ee_{s,0,1}[
e^{-\alpha\tau^*}(K-S_{\tau^*})^+\,\mathbf{1}_{\{\tau_{s_0}<\infty\}}|\mathcal{F}_{\tau_{s_0}}]
\right]\\
=&\ee_{s,0,1}
\left[\rule{0pt}{12pt}\right.
e^{-\alpha\tau_{s_0}}\mathbf{1}_{\{\tau_{s_0}< \infty\}}
\underbrace{
\ee_{s_0,1,1}[e^{-r\tau^*}(K-S_{\tau^*})^+]
}_{V(s_0,1,1)}
\left.\rule{0pt}{12pt}\right],
\end{align*}
by strong Markov property.
So,
simply working out the Laplace transform of the hitting time $\tau_{s_0}$
yields (v1).

Next we verify (v2), for $s>s_0$ fixed.
By Markov Property, we only need to show that
\begin{equation}\label{for supermarti}
\ee_{s,0,1}[e^{-\alpha t}V^*(S_t,Y_t,\eta_t)]\leq V^*(s,0,1),
\end{equation}
for all $t\ge 0$.

Fix $t\ge 0$, and consider
\begin{align}\label{equality for supermart}
e^{-\alpha t}\,V^*(S_t,Y_t,\eta_t)
\;&=\;
e^{-\alpha t}\,V^*(S_t,Y_t,\eta_t)
-
e^{-\alpha(t\wedge\tau_{s_0})}\,
V^*(S_{t\wedge\tau_{s_0}},Y_{t\wedge\tau_{s_0}},\eta_{t\wedge\tau_{s_0}})\nonumber\\
\;&+\;
e^{-\alpha(t\wedge\tau_{s_0})}\,V^*(S_{t\wedge\tau_{s_0}},0,1),
\end{align}
where the last term is justified by $V^*(s_0,1,1)=V^*(s_0,0,1)$.

Now realise that, by strong Markov property,
\begin{align*}
&\;\ee_{s,0,1}[\,{\bf 1}_{\{t\ge\tau_{s_0}\}}\,e^{-\alpha t}\,V^*(S_t,Y_t,\eta_t)\,]\\
\rule{0pt}{20pt}
=\;&\;
\ee_{s,0,1}
\left[\rule{0pt}{12pt}\right.
{\bf 1}_{\{t\ge\tau_{s_0}\}}
\ee_{s,0,1}[\,e^{-\alpha t}\,V^*(S_t,Y_t,\eta_t)
\,|\,{\cal F}_{t\wedge\tau_{s_0}}\,]
\left.\rule{0pt}{12pt}\right]\\
=\;&\;
\int{\bf 1}_{\{t\ge\tau_{s_0}(\omega)\}}\,e^{-\alpha\tau_{s_0}(\omega)}\,
\underbrace{
\ee_{s_0,1,1}[\,e^{-\alpha(t-\tau_{s_0}(\omega))}\,
V^*(S_{t-\tau_{s_0}(\omega)},Y_{t-\tau_{s_0}(\omega)},\eta_{t-\tau_{s_0}(\omega)})\,]
}_{
\mbox{\small\hspace{4cm}$\le\,V^*(s_0,1,1)$ from case (ii)}
}
\,\pp_{s,0,1}(\dd\omega),
\end{align*}
and, since the difference on the right-hand side of (\ref{equality for supermart})
equals
\[
\left[\rule{0pt}{12pt}\right.
e^{-\alpha t}\,V^*(S_t,Y_t,\eta_t)
-
e^{-\alpha\tau_{s_0}}\,V^*(s_0,1,1)]
\left.\rule{0pt}{12pt}\right]
\times
{\bf 1}_{\{t\ge\tau_{s_0}\}},
\]
we obtain that
\begin{equation}\label{ineq for supermart}
\ee_{s,0,1}[\,e^{-\alpha t}\,V^*(S_t,Y_t,\eta_t)\,]
\,\le\,
\ee_{s,0,1}[\,
e^{-\alpha(t\wedge\tau_{s_0})}\,V^*(S_{t\wedge\tau_{s_0}},0,1)\,].
\end{equation}

To prove (\ref{for supermarti}), 
we want to apply It\^o's formula on the above right-hand side
followed by taking expectations.

Recall the operator $L$ introduced on page \pageref{big generator}.
As the function $V^*(\cdot,0,1)$ defined on $(s_0,\infty)$ can be extended 
to a $C^2$-function on $\bfR$, It\^o's formula $\pp_{\!\!s,0,1}$-a.s.\ yields
\[
e^{-\alpha(t\wedge\tau_{s_0})}\,V^*(S_{t\wedge\tau_{s_0}},0,1)
\,=\,
V^*(s,0,1)\,+
\int_{0}^{t\wedge\tau_{s_0}}
e^{-\alpha u}(L-\alpha I)V^*(S_u,0,1)\,\text{d}u\,+\,I_{BM},
\]
where $I_{BM}$ is an integrable stochastic integral against Brownian motion
whose expectation vanishes.

Furthermore,
by explicit calculation, $(L-\alpha I)V^*(s',0,1)=0$, for $s'>s_0$, 
and hence the right-hand side of (\ref{ineq for supermart})
reduces to $V^*(s,0,1)$ eventually showing
(\ref{for supermarti}).

It remains to verify (v3), for any $s>s_0$, in both cases (a) and (b).

For (a), observe that
\[
V^*(s,0,1)\,=\,V(s_0,1,1)\left(\frac{s}{s_0}\right)^{\gamma^-}
\,\ge\,\;
(K-b_0)\left(\frac{s_0}{b_0}\right)^{\gamma^-}
\hspace{-5pt}
\left(\frac{s}{s_0}\right)^{\gamma^-}
\,=\,\;(K-b_0)\left(\frac{s}{b_0}\right)^{\gamma^-},
\] 
and hence it suffices to show that 
$(K-b_0)(s/b_0)^{\gamma^-}\ge(K-s_0)^+$, for $s>s_0$.
Note that $(K-b_0)(\cdot/b_0)^{\gamma^-}$ coincides, on $[b_0,\infty)$,
with the value function $\tilde{V}$ given in Theorem \ref{18}
which satisfies both $\tilde{V}(s)\ge(K-s)^+$, for $s\in[b_0,\infty)$,
and $\tilde{V}'(b_0)=-1$. Therefore, because
$(K-b_0)(\cdot/b_0)^{\gamma^-}$ is a convex function on
$(0,\infty)$, it must be bounded below by $(K-\cdot)^+$ on
$(s_0,b_0)$, too.

For (b), there is nothing to show, if $s_0\ge K$.
But, if $b_0\le s_0<K$, then we know from the proof of Lemma \ref{23} that
$(K-s_0)(s/s_0)^{\gamma^-}>(K-s)^+$, for all $s>s_0$. Thus,
\[
V^*(s,0,1)\,=\,V(s_0,1,1)\left(\frac{s}{s_0}\right)^{\gamma^-}>\,\;(K-s)^+,
\quad\mbox{for}\; s>s_0,
\]
because, in case (b), $V(s_0,1,1)>(K-s_0)^+$ by assumption.
\subsubsection{Proof of Theorem \ref{main} (iii)(c)}
We are going to verify (v1),(v2),(v3) from Section \ref{24},
for fixed $s>s_0$, but using $\tau^*=\tau_{s_0}$.

First, note that, on $\{\tau_{s_0}<\infty\}$, the process
$(S_t, 0\leq t\leq \tau_{s_0})$ under $\pp_{\!\!s,0,1}$ has 
the same distribution as 
the process $(\tilde{S}_t, 0\leq t\leq \tilde{\tau}_{s_0})$ 
under $\tilde{\pp}_{\!\!s}$ introduced in Theorem \ref{18}.
By Lemma \ref{23}, for $s>s_0$, we therefore have
\begin{align*}
&\ee_{s,0,1} [e^{-\alpha\tau_{s_0}}(K-S_{\tau_{s_0}})^+\,
{\bf 1}_{\{\tau_{s_0}<\infty\}}]\\
=\;&
\tilde{\ee}_s [e^{-\alpha\tilde{\tau}_{s_0}}(K-\tilde{S}_{\tilde{\tau}_{s_0}})^+\,
{\bf 1}_{\{\tilde\tau_{s_0}<\infty\}}]
\,=\,
(K-s_0)\bigg(\frac{s}{s_0}\bigg)^{\gamma^-},
\end{align*}
where $V(s_0,1,1)=K-s_0$ by assumption, and (v1) follows.

For (v2), one can copy the corresponding proof in Section \ref{24}
because the value function has  the same form in all sub\,-cases (a,b,c).

Finally, $V^*(s,0,1)=\tilde{V}_0(s)\ge(K-s)^+$, for all $s>s_0$,
showing (v3) and finishing the proof of Theorem \ref{main}(iii).
\subsection{Proof of Theorem \ref{main} (iv)}
As in the proof of Theorem \ref{main}(ii), we first guess the
structure of the stopping region and then verify that the solution of
the corresponding free-boundary value problem is the wanted value
function.

First, as the value function is claimed to be attained at
the Markov time
$\tau_{[b_*,b_0],0}\wedge\tau_{b_1}\wedge\tau_{\eta,0}$,
the stopping region to be guessed should take the form
\[
\left[\rule{0pt}{12pt}\right.
[b_*,b_0]\times\{0\}\times\{1\}
\left.\rule{0pt}{12pt}\right]
\cup
\left[\rule{0pt}{12pt}\right.
(0,b_1]\times \{1\}\times\{1\}
 \left.\rule{0pt}{12pt}\right]
\cup
 \left[\rule{0pt}{12pt}\right.
(0,\infty)\times\{1\}\times\{0\}
\left.\rule{0pt}{12pt}\right],
\]
where $b_0,\,b_1$ are already known, but $b_*\in[s_0,b_0)$ is not.
Recall that both $b_0$ and $b_1$ must be dominated by $K$.

Now, referring to the proof of Theorem \ref{main}(ii) for the
underlying argument, the corresponding free-boundary value problem
for the unknown value function $V(s,0,1)$ is
\[
0\,=\,
\mu_0 s\partial_{1}V(s,0,1)+
\frac{1}{2}\sigma_0^2 s^2\partial_{11}V(s,0,1)
-\alpha V(s,0,1),  \quad\text{ for }s \in (s_0,b_*)\cup(b_0,\infty),
\]
depending on the unknown $b_*$ subject to the boundary conditions
\begin{equation}\label{29}
\begin{array}{rcl}
V(s_0,1,1)&=&V(s_0+,0,1),\\
\rule{0pt}{15pt}
V(b_{*}-,0,1)&=&K-b_*,\\
\rule{0pt}{15pt}
\partial_1 V(b_{*}-,0,1)&=&-1,\\
\rule{0pt}{15pt}
K-b_0&=&V(b_{0}+,0,1),\\
\rule{0pt}{15pt}
-1&=&\partial_1 V(b_{0}+,0,1),\\
\rule{0pt}{15pt}
\lim_{s\to\infty}V(s,0,1)&=&0.
\end{array}
\end {equation}

Taking into account Theorem \ref{18}, if $V(s,0,1)$ satisfies
these constraints, then it must have the representation
\begin{equation}
\begin{cases}
e_1^* s^{\gamma^+}+\;e_2^* s^{\gamma^-} &:\quad s\in  (s_0,b_*)\\
\rule{0pt}{20pt}
K-s &:\quad s\in [b_*, b_0]\cap(s_0,\infty)\\
(K-b_0)\bigg(\dfrac{s}{b_0}\bigg)^{\gamma^-} &:\quad s\in(b_0,\infty) \\
\end{cases}\label{31}
\end{equation}
where $e_1^*,\,e_2^*,\,b_*$ should be determined by the first three conditions
of (\ref{29}).
However, these three conditions might not make $e_1^*,\,e_2^*,\,b_*$ unique.
But, $V$ is the value function associated with an optimal stopping problem,
and this leads to the uniqueness 
stated in the next lemma---see Remark \ref{unique b*}(ii) below.
\begin{lemma}\label{choice b*}
If $s_0<b_0$ and
$(K-s_0)<V(s_0,1,1)< (K-b_0)({s_0}/{b_0})^{\gamma^-}$, 
then there exist
unique coefficients $e_1^*,\,e_2^*$ 
and a unique stopping level $b_*\in(s_0,b_0)$ such that
\begin{align}
V(s_0,1,1)&\,=\,
e_1^* s_0^{\gamma+}+\;e_2^* s_0^{\gamma^-},\label{1000}\\
\rule{0pt}{15pt}
e_1^* b_*^{\gamma^+}+\;e_2^* b_*^{\gamma^-}&\,=\,K-b_*,\label{40}\\
\rule{0pt}{15pt}
e_1^*\gamma^+ b_*^{\gamma^+}+\;e_2^*\gamma^- b_*^{\gamma^-}&\,=\,-b_*,\nonumber\\
\rule{0pt}{15pt}
e_1^* s^{\gamma^+}+\;e_2^* s^{\gamma^-}&\,>\,
K-s, \quad \text{for } s\in (s_0,b_*).\nonumber
\end{align}
\end{lemma}

Before proving this lemma, we are going to state the following 
preparatory results.
\begin{lemma}(See \cite[Lemma 2]{lehoczky} for example.)\label{hitting times}
For fixed $s_0\leq s\leq\tilde{s}$,
\begin{equation}
\phi_1(s,\tilde{s})\stackrel{\mbox{\tiny def}}{=}
\ee_{s,0,1}[e^{-\alpha(\tau_{s_0}\wedge \tau^{\tilde{s}})}
\mathbf{1}_{\{\tau_{s_0}<\tau^{\tilde{s}}\}}]
=
\dfrac{s^{\gamma^+}\tilde{s}_{\rule{0pt}{3pt}}^{\,\gamma^-}
-s^{\gamma^-}\tilde{s}_{\rule{0pt}{3pt}}^{\,\gamma^+}}
{s_0^{\gamma^+}\tilde{s}_{\rule{0pt}{3pt}}^{\,\gamma^-}
-s_0^{\gamma^-}\tilde{s}_{\rule{0pt}{3pt}}^{\,\gamma^+}},\nonumber
\end{equation} and 
\begin{equation}
\phi_2(s,\tilde{s})\stackrel{\mbox{\tiny def}}{=}
\ee_{s,0,1}[e^{-\alpha(\tau_{s_0}\wedge \tau^{\tilde{s}})}
\mathbf{1}_{\{\tau_{s_0}>\tau^{\tilde{s}}\}}]
=
\dfrac{s_0^{\gamma^+}s_{\rule{0pt}{3pt}}^{\gamma^-}
-s_0^{\gamma^-}s_{\rule{0pt}{3pt}}^{\gamma^+}}
{s_0^{\gamma^+}\tilde{s}_{\rule{0pt}{3pt}}^{\,\gamma^-}
-s_0^{\gamma^-}\tilde{s}_{\rule{0pt}{3pt}}^{\,\gamma^+}},\nonumber
\end{equation}
where
\begin{equation*}
\tau^{\tilde{s}} \stackrel{\mbox{\tiny def}}{=} \inf\{t\ge 0:S_t \geq \tilde{s} \}.
\end{equation*}
\end{lemma}
\begin{lemma}\label{nullstellen}
Under the assumptions of Lemma \ref{choice b*}, the equation
\begin{equation}
V(s_0,1,1)\bigg(\dfrac{s}{s_0}\bigg)^{\gamma^-}=\;K-s\nonumber 
\end{equation}
has exactly two solutions $s_1,s_2\in(s_0,K)$
only one of which, say $s_1$, is less than $b_0$.
\end{lemma}
\begin{proof}
Consider the function
$f(\cdot)= V(s_0,1,1)(\cdot/s_0)^{\gamma^-}-(K-\cdot)$ on $(0,\infty)$,
and note that $f(s_0)>0$, $f(b_0) <0$, $f(K)>0$. 
By Intermediate Value Theorem, 
there exist $s_0<s_1<b_0$ and $b_0<s_2<K$ such that
$f(s_1)=f(s_2)=0$, that is,
\begin{equation*}
V(s_0,1,1)\bigg(\dfrac{s_i}{s_0}\bigg)^{\gamma^-}=\;K-s_i,
\quad\text{for }i=1,2.\label{39}
\end{equation*}
There exist exactly two points, only, because 
$V(s_0,1,1)(\cdot/s_0)^{\gamma^-}$ is strictly convex on $(0,\infty)$ 
and can be intersected by a line at no more than two points. 
Moreover,  
\begin{equation}
V(s_0,1,1)\bigg(\dfrac{s}{s_0}\bigg)^{\gamma^-}<\;K-s,\label{1} 
\quad\text{for } s_1<s<s_2,
\end{equation}
because of strict convexity, too.
\end{proof}
\begin{proof}[Proof of {\bf Lemma \ref{choice b*}}]
Introduce
\begin{equation}
\Gamma(s,\tilde{s})\stackrel{\mbox{\tiny def}}{=}
V(s_0,1,1)\phi_1(s,\tilde{s})+(K-\tilde{s})\phi_2(s,\tilde{s}),
\quad\mbox{for} s_0\le s\le\tilde{s},\nonumber
\end{equation} 
using the functions $\phi_1,\,\phi_2$ defined in Lemma \ref{hitting times}.
For fixed $\tilde{s}>s_0$, note that the function
$[s_0,\tilde{s}\,]\ni s\mapsto\Gamma(s,\tilde{s})$
is of the form $e_1\,s^{\gamma^+}\!+e_2\,s^{\gamma^-}$ with
\begin{equation*}
e_1\,=\,
\frac{V(s_0,1,1)\tilde{s}_{\rule{0pt}{3pt}}^{\,\gamma^-}
-(K-\tilde{s})s_0^{\gamma^-}}
{s_0^{\gamma^+}\tilde{s}_{\rule{0pt}{3pt}}^{\,\gamma^-}
-s_0^{\gamma^-}\tilde{s}_{\rule{0pt}{3pt}}^{\,\gamma^+}},
\quad
e_2\,=\,
\frac{-V(s_0,1,1)\tilde{s}_{\rule{0pt}{3pt}}^{\,\gamma^+}
+(K-\tilde{s})s_0^{\gamma^+}}
{s_0^{\gamma^+}\tilde{s}_{\rule{0pt}{3pt}}^{\,\gamma^-}
-s_0^{\gamma^-}\tilde{s}_{\rule{0pt}{3pt}}^{\,\gamma^+}},
\end{equation*}
and that both boundary conditions
\[
e_1\,s_0^{\gamma^+}+\,e_2\,s_0^{\gamma^-}
=\,
\Gamma(s_0,\tilde{s})
\,=\,
V(s_0,1,1)
\quad\&\quad
e_1\,\tilde{s}_{\rule{0pt}{3pt}}^{\,\gamma^+}
+\,e_2\,\tilde{s}_{\rule{0pt}{3pt}}^{\,\gamma^-}
=\,
\Gamma(\tilde{s},\tilde{s})
\,=\,
K-\tilde{s}
\]
are satisfied.
Hence, if we can show that there is exactly one $b_*\in(s_0,b_0)$ such that both
\[
\partial_1\Gamma(b_*,b_*)=-1
\quad\mbox{and}\quad
\Gamma(s,b_*)>K-s,\;s\in(s_0,b_*),
\]
then the triplet $(e_1^*,e_2^*,b_*)$, 
where $e_1^*,e_2^*$ are given by the above formulae for $e_1,e_2$
when replacing $\tilde{s}$ by $b_*$,
would be the unique solution 
of the problem stated in Lemma \ref{choice b*}.
Here the uniqueness of $e_1^*,e_2^*$ follows from the uniqueness of $b_*$
as the formulae for $e_1^*,e_2^*$ coincide with the unique solution to
the sub\,-system (\ref{1000}),(\ref{40}) of the conditions in Lemma \ref{choice b*} 
when treating $s_0^{\gamma^\pm}$ and $b_*^{\gamma^\pm}$ as coefficients.

First, we show that that there is $b_*\in(s_0,b_0)$ such that
$\partial_1\Gamma(b_*,b_*)=-1$. For the uniqueness of $b_*$
we refer to Remark \ref{unique b*}(ii) below.

Using simple calculations based on It\^o's formula,
observe that, for any $s\ge s_0$, the stochastic process
$(e^{-\alpha(t\wedge \tau_{s_0})}
V(s_0,1,1)(S_{t\wedge \tau_{s_0}}/s_0)^{\gamma^-},\,t\geq 0)$
is a $\pp_{\!\!s,0,1}$-martingale.
Now, recall $s_1$ from Lemma \ref{nullstellen} and the stochastic representation
of $\Gamma(s,\tilde{s})$ in terms of $\phi_1,\phi_2$.
Then, by Doob's Optional Sampling Theorem, for $s_0\leq s\leq s_1$,
\begin{align*}
V(s_0,1,1)\bigg(\dfrac{s}{s_0}\bigg)^{\gamma^-} 
&=\;\ee_{s,0,1}[\;e^{-\alpha(\tau_{s_0}\wedge\tau^{s_1})}
V(s_0,1,1)\bigg(\dfrac{S_{\tau_{s_0}\wedge\tau^{s_1}}}{s_0}\bigg)^{\gamma^-}\!]\nonumber\\
&=\; V(s_0,1,1)\phi_1(s,s_1)+V(s_0,1,1)\bigg(\dfrac{s_1}{s_0}\bigg)^{\gamma^-}\phi_2(s,s_1)\\
&=\;V(s_0,1,1)\phi_1(s,s_1)+(K-s_1)\phi_2(s,s_1)\\
\rule{0pt}{20pt}
&=\;\Gamma(s,s_1),
\end{align*} 
using Lemma \ref{nullstellen} to justify the penultimate equality above.
Thus,
\begin{equation}\label{derivative at s_1}
\partial_1\Gamma(s_1,s_1)
\,=\,
\partial_s\left[
V(s_0,1,1)\bigg(\dfrac{s}{s_0}\bigg)^{\gamma^-}
\right]_{s=s_1}.
\end{equation}

Next, for any $s>s_0$,
\[
\partial_1\Gamma(s,s)
\,=\,
\frac{
V\,s^{\gamma^+ +\gamma^- -1}\,(\gamma^+ -\gamma^-)
+(K-s)\,s^{-1}\,(
\gamma^- s_0^{\gamma^+}{s}_{\rule{0pt}{3pt}}^{\,\gamma^-}
-\gamma^+ s_0^{\gamma^-}{s}_{\rule{0pt}{3pt}}^{\,\gamma^+}
)
}{
s_0^{\gamma^+}{s}_{\rule{0pt}{3pt}}^{\,\gamma^-}
-\,s_0^{\gamma^-}{s}_{\rule{0pt}{3pt}}^{\,\gamma^+}
},
\]
where $V$ stands for $V(s_0,1,1)$. Consider the above right-hand side
as a function of $(V,s)$ which we denote by $g(V,s)$ in what follows.

Choose $s=b_0$ and realise that the function $V\mapsto g(V,b_0)$
is strictly decreasing, since $s_0<b_0$ implies
$s_0^{\gamma^+}{b_0}_{\rule{0pt}{3pt}}^{\,\gamma^-}
-\,s_0^{\gamma^-}{b_0}_{\rule{0pt}{3pt}}^{\,\gamma^+}<0$.
Moreover $g((K-b_0)(s_0/b_0)^{\gamma^-},b_0)=-1$, so that
our assumption of
$V(s_0,1,1)< (K-b_0)(s_0/b_0)^{\gamma^-}$
implies
$g(V(s_0,1,1),b_0)>-1$.

But, using (\ref{derivative at s_1}), we also have that
\[
g(V(s_0,1,1),s_1)
\,=\,
\partial_s\left[
V(s_0,1,1)\bigg(\dfrac{s}{s_0}\bigg)^{\gamma^-}
\right]_{s=s_1},
\]
where, by the same arguments used in the proof of Lemma \ref{nullstellen},
the above right-hand side must be less than $-1$.

All in all, we obtain that $g(V(s_0,1,1),s_1)<-1<g(V(s_0,1,1),b_0)$.
And since the function $g(V(s_0,1,1),\cdot)$ 
is continuous on $(s_0,\infty)$,
and since $s_0<s_1<b_0$ by Lemma \ref{nullstellen}, 
it is again a consequence of the Intermediate Value Theorem that
there exists $b_*\in(s_1,b_0)$ such that 
$\partial_1\Gamma(b_*,b_*)=g(V(s_0,1,1),b_*)=-1$.

Second, finishing the proof of Lemma \ref{choice b*},
we will show that
\[
\Gamma(s,b_*)\,=\,
e_1^* s^{\gamma^+}+\,e_2^* s^{\gamma^-}
>\,K-s,
\quad \text{for } s\in (s_0,b_*),
\]
where
\[
e_1^*\,=\,
\frac{V(s_0,1,1)b_*^{\,\gamma^-}
-(K-b_*)s_0^{\gamma^-}}
{s_0^{\gamma^+}b_*^{\,\gamma^-}
-s_0^{\gamma^-}b_*^{\,\gamma^+}},
\quad
e_2^*\,=\,
\frac{-V(s_0,1,1)b_*^{\,\gamma^+}
+(K-b_*)s_0^{\gamma^+}}
{s_0^{\gamma^+}b_*^{\,\gamma^-}
-s_0^{\gamma^-}b_*^{\,\gamma^+}}.
\]

In order to do so, we analyse the function
$f(s)=e_1^*\,s^{\gamma^+}\!+e_2^*\,s^{\gamma^-},\,s>0$.
Since $f'(b_*)=-1$ has already been shown,
$f(s)>K-s,\,s\in(s_0,b_*)$, would follow from $f$ being strictly convex
on $(0,b_*)$ which we are going to prove below.

First, observe that both $e_1^*$ and $e_2^*$ are positive. In fact,
as $s_0<b_*$ implies that the denominator 
$s_0^{\gamma^+}b_*^{\,\gamma^-}\!-s_0^{\gamma^-}b_*^{\,\gamma^+}$
is negative,
the positivity of the two coefficients $e_1^*,e_2^*$ follows from
\[
V(s_0,1,1)b_*^{\,\gamma^-}-(K-b_*)s_0^{\gamma^-}<0
\quad\mbox{and}\quad
-V(s_0,1,1)b_*^{\,\gamma^+}+(K-b_*)s_0^{\gamma^+}<0,
\]
where the former inequality is a consequence of (\ref{1}),
because $b_*$ was chosen from the interval $(s_1,b_0)$,
while the latter inequality is a consequence of our assumption
$K-s_0<V(s_0,1,1)$, on the one hand hand, and 
$s_0<b_*,\,\gamma^+>0$, on the other.

As a consequence, if $\gamma^+\ge 1$, then $f$ is the sum of two
strictly convex functions, and hence strictly convex everywhere.

Now, assume $0<\gamma^+<1$ which is the remaining case. 
Note that $f$ has exactly one local minimum at 
$s'=[\frac{-e_2^*\gamma^-}{e_1^*\gamma^+}]^{\frac{1}{\gamma^+-\gamma^-}}$,
and that this minimum is global because 
$f(0+)=\lim_{s\to\infty}f(s)=\infty$.

Thus, $f$ is strictly decreasing on $(0,s')$ and strictly increasing
on $(s',\infty)$ which implies $b_*<s'$ because $f'(b_*)=-1$.

Finally, $f$ has exactly one point of inflection at
$s''=[
\frac{-e_2^*\gamma^-(\gamma^- -1)}{e_1^*\gamma^+(\gamma^+ -1)}
]^{\frac{1}{\gamma^+-\gamma^-}}$,
and this point satisfies $s''>s'$. 
Hence,
$f$ must be at least strictly convex on $(0,s')$ 
which also proves its strict convexity on $(0,b_*)\subseteq(0,s')$.
\end{proof}
\begin{remark}\rm\label{unique b*}
\begin{itemize}
\item[(i)]
In the case of $K-s_0=V(s_0,1,1)$,
the choice of $b_*$ has not been discussed yet. In this case, 
we claim that $b_*=s_0$,
and we will verify below that the corresponding function 
given by (\ref{31}) coincides with $V(\cdot,0,1)$ on $(s_0,\infty)$.
\item[(ii)]
In the case of
$(K-s_0)<V(s_0,1,1)< (K-b_0)(s_0/b_0)^{\gamma^-}$,
we showed existence of $b_*\in(s_0,b_0)$,
and any such $b_*$ uniquely determines a function as given by (\ref{31}).
We will verify below that any such function 
coincides with $V(\cdot,0,1)$ on $(s_0,\infty)$.
Moreover, as shown in the proof of Lemma \ref{choice b*},
for any choice of $b_*$, the corresponding function given by (\ref{31}) must be
strictly convex on $(s_0,b_*)$. 
Hence, as the value function with respect to an optimal stopping problem is unique,
there can only be one $b_*$.
\end{itemize}
\end{remark}

Set $\tau^*=\tau_{[b_*,b_0],0}\wedge\tau_{b_1}\wedge\tau_{\eta,0}$,
and introduce:
\[
V^*(s,y,i)\,\stackrel{\mbox{\tiny def}}{=}\,
\left\{\begin{array}{lcl}
e_1^* s^{\gamma^+}+\;e_2^* s^{\gamma^-} 
&:&  s\in  (s_0,b_*),y=0,i=1\\
\rule{0pt}{20pt}
K-s &:&  s\in [b_*, b_0]\cap(s_0,\infty),y=0,i=1\\
(K-b_0)\left(\dfrac{s}{b_0}\right)^{\gamma^-} 
&:& s\in(b_0,\infty),y=0,i=1 \\
\rule{0pt}{15pt}
V(s,1,i)&:&s>0,y=1,i\in\{0,1\}
\end{array}\right.
\]
Again, by verifying  the conditions (v1),(v2),(v3) stated 
at the beginning of Section \ref{24}
for any fixed $s>s_0$,
we complete both the program set out in Remark \ref{unique b*}
and the proof of Theorem \ref{main}(iv).

For (v1), if $s\in[b_*,b_0]$, then there is nothing to prove,
as one stops immediately, and if $s>b_0$, then 
(v1) follows from Theorem \ref{18}, as $V^*(\cdot,0,1)$ coincides 
with $\tilde{V}$ on $(b_0,\infty)$. 

In the remaining case of $s_0<s<b_*$, observe that the assumptions
of Lemma \ref{choice b*} are satisfied, as $K-s_0=V(s_0,1,1)$ can be ruled out.
Furthermore,
\begin{align*}
&\ee_{s,0,1}[e^{-\alpha\tau_{[b_*,b_0],0}\wedge\tau_{b_1}\wedge\tau_{\eta,0}}
(K-S_{\tau_{[b_*,b_0],0}\wedge\tau_{b_1}\wedge\tau_{\eta,0}})^+]\\
=&
\ee_{s,0,1}[e^{-\alpha(\tau_{b_1}\wedge\tau_{\eta,0})}
(K-S_{\tau_{b_1}\wedge\tau_{\eta,0}})^+
\mathbf{1}_{\{\tau_{[b_*,b_0],0}>\tau_{s_0}\}}]
+
\ee_{s,0,1}[e^{-\alpha\tau_{b_*}}(K-b_*)^+\mathbf{1}_{\{\tau_{[b_*,b_0],0}<\tau_{s_0}\}}]\\
=&
\ee_{s,0,1}
\left[\rule{0pt}{12pt}\right.
\ee_{s,0,1}[\,e^{-\alpha(\tau_{b_1}\wedge\tau_{\eta,0})}
(K-S_{\tau_{b_1}\wedge\tau_{\eta,0}})^+
\mathbf{1}_{\{\tau_{[b_*,b_0],0}>\tau_{s_0}\}}|\mathcal{F}_{\tau_{s_0}}\,]
\left.\rule{0pt}{12pt}\right]
+
(K-b_*)\,\phi_2(s,b_*)
\end{align*}
which, by strong Markov property, simplifies to
\[
V(s_0,1,1)\,\phi_1(s,b_*)
+
(K-b_*)\,\phi_2(s,b_*).
\]
Using the definition of $\Gamma$ given at the beginning
of the proof of Lemma \ref{choice b*} on page \pageref{choice b*},
the last expression equals $\Gamma(s,b_*)$.
And since $b_*$ was chosen such that
$\Gamma(s,b_*)=V^*(s,0,1)$, condition (v1)
follows in the remaining case, too.

For (v2), realise that,
by the same arguments used in Section \ref{24} on page \pageref{ineq for supermart},
the above candidate value function of this section satisfies 
(\ref{ineq for supermart}),
and hence we only need to show that
\begin{equation}\label{sufficient for v2}
\ee_{s,0,1}[\,
e^{-\alpha(t\wedge\tau_{s_0})}\,V^*(S_{t\wedge\tau_{s_0}},0,1)\,]
\,\le\,
V^*(s,0,1),
\end{equation}
for any $t\ge 0$.

Now, consider the function
$\phi(s')=V^*(s',0,1),\,s'>s_0$.
Note that $\phi$ can be extended to a function on $(0,\infty)$
which is of the type described in Remark \ref{our ito}
having finitely many exceptional points of insufficient smoothness.
In the case of $K-s_0=V(s_0,1,1)$, where $b_*=s_0$, 
there is one exceptional point at $b_0$,
whereas in the case of $K-s_0<V(s_0,1,1)$ there are two exceptional points
at $b_*,b_0$.
However, in both cases, 
applying Theorem \ref{18} and Lemma \ref{choice b*}, respectively,
one can extend $\phi$ in such a way that $\phi'$ is continuous.

Thus, according to Remark \ref{our ito}, 
we $\pp_{\!\!s,0,1}$-a.s.\ have
\begin{equation}\label{ito case iv}
e^{-\alpha(t\wedge\tau_{s_0})}V^*(S_{t\wedge\tau_{s_0}},0,1)
\,=\,
V^*(s,0,1)+\int^{t\wedge\tau_{s_0}}_0
e^{-\alpha u}(L-\alpha I)V^*(S_u,0,1)\text{d}u+I_{BM},
\end{equation}
where $L$ stands once more for the operator 
introduced on page \pageref{big generator},
and $I_{BM}$ is an integrable stochastic integral against Brownian motion
whose expectation vanishes.
 
Next, by explicit calculation, 
$(L-\alpha I)V^*(\cdot,0,1)\le 0$ on $(s_0,b_*)\cup(b_*,b_0)\cup(b_0,\infty)$.
Since the process $(S_u,\,u\in[0,\tau_{s_0}))$ has $\pp_{\!\!s,0,1}$\,-\,a.s.\
no occupation time in $\{b_*,b_0\}$,
inequality (\ref{sufficient for v2}) follows from (\ref{ito case iv}) 
by taking expectations, proving (v2).

Finally, condition (v3) is a consequence 
of Lemma \ref{choice b*}, if $s\in(s_0,b_*)$,
and of Theorem \ref{18}, if $s\in(b_0,\infty)$.
Otherwise, there is nothing to show.
\section {Numerical Analysis and Discussion}\label{numerical}
We are going to discuss the four cases (iii)(a-c) and (iv) of Theorem \ref{main}
using practically relevant values for 
$s_0,\mu_0,\sigma_0,\mu_1$, $\sigma_1,\alpha,\lambda,K$.

Note that the choice of $\mu_1,\sigma_1,\alpha,\lambda,K$ fixes the value of $b_1$,
and that the two cases (iii)(b,c) of Theorem \ref{main}
can be reformulated as
\begin{itemize}
\item[($b'$)]
$b_0\leq s_0$ and $b_1<s_0$,
\item[($c'$)]
$b_0\leq s_0<K$ and $b_1\ge s_0$.
\end{itemize}
However, the formulation of the two cases (iii)(a) and (iv)
requires the value of $V(s_0,1,1)$, and that's why we decided
to formulate (iii)(b,c) using $V(s_0,1,1)$, too.

In what follows, when using the noun `put' without further specification,  
we mean a {\it perpetual} American put as considered in Theorem \ref{main}.
However, as motivated in Remark \ref{world}(ii,iii,v),
the average length of the put's optimal exercise time is supposed to be rather short,
and hence we think that our analysis also produces
good benchmarks for traded American puts with times to maturity being long enough to allow
for medium term option trading, that is, three months and longer.

First, we have to choose the put's underlying asset.
By the macroeconomic explanation given by Black in \cite{black1976},
we think that a leverage effect is more likely to be observed
when the whole market falls, and hence we choose an index,
say, the Dow Jones index.

Second, to fully determine the put, a strike level has to be chosen.
At the end of this section we give a summary of how to choose the strike level
motivated by our discussion of Theorem \ref{main} below.
For now we choose $K=17000$ for demonstration.

Next we fix the following hypothetical values for
$\sigma_0=20\%,\,\mu_1=0,\,\sigma_1=35\%,\,\alpha=5\%,\,\lambda=100$,
and refer to Remark \ref{world} and Remark \ref{discount}(i) for their interpretation.

So, $\sigma_0=20\%$ is supposed to be the implied volatility at present time
of a traded American put with strike $K$ and time to maturity of at least three months,
and we assume that the expected market drop would cause an 
`excited' volatility of $\sigma_1=35\%$.

Setting $\lambda=100$ means to assume that the `excited' state
would only last for half a week on average.

After the market has dropped, it is not clear what the new trend $\mu_1$
of the index would be. Furthermore, if the `excited' state only lasts
for a short period of time, one can assume that the `excited' fluctuations
according to the bigger $\sigma_1$ dominate the trend.
Thus a reasonable choice for the new trend would be $\mu_1=0$.

The two remaining non-fixed parameters are $s_0$ and $\mu_0$.
Since $\sigma_0,\alpha,K$ have been fixed, there is a one-to-one
correspondence between $\mu_0$ and $b_0$, and hence,
each pair $(\mu_0,s_0)$ determines one of the four cases
(iii)(a-c) and (iv) of Theorem \ref{main}.
The optimal stopping rules given in each of these cases
are called {\em strategies} of the trader, in what follows.

In practice, depending on the present value $s$ of the Dow Jones, 
the trader would choose $s_0$ according to their preferences
of the future---they expect a market drop of a certain size.
In our analysis we take the reverse point of view:
we first classify the values of $s_0$ and then discuss the impact
of present values $s$ \underline{above} $s_0$
on the strategy to be chosen by the trader.

While $s-s_0$ determines the size of the expected market drop,
the choice of $\mu_0$ determines how soon this is supposed to happen
in terms of the model (recall that $\sigma_0$ has been fixed).
A rather small value of $\mu_0$ should be used if one wants that
many price-trajectories predicted by the model reach the level $s_0$
in a rather short time. For example, a value of $\mu_0=-100\%$ would imply
that, roughly, the value of the index expected under the model
drops from 15600 to 15000 within two weeks.

In contrast, in the case of bigger values of $\mu_0$, the model
more often predicts rising values of the index in the future,
and this is of course not in accordance with an expected market drop.
We will nevertheless analyse bigger values of $\mu_0$ 
because the corresponding strategies might be of use for the trader 
in case they learn during the trade 
that their preferences of the future were wrong.
 
Figure 1 below shows the blue graph of $b_0=b_0(\mu_0)$ embedded into 
the $(\mu_0,s_0)$\,-\,plane.
The red horizontal line marks the level $b_1=14658$ 
which crosses $b_0(\mu_0)$ at $\mu_0=13.7\%$.
\begin{center}
\includegraphics[width=7.5cm]{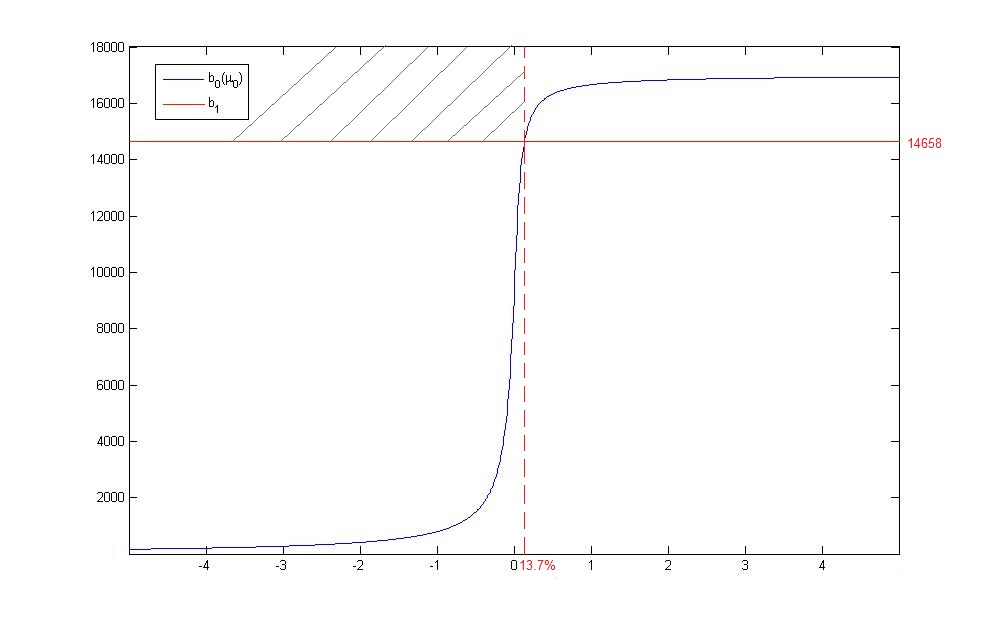}
\hspace{-.5cm}
\includegraphics[width=7.5cm]{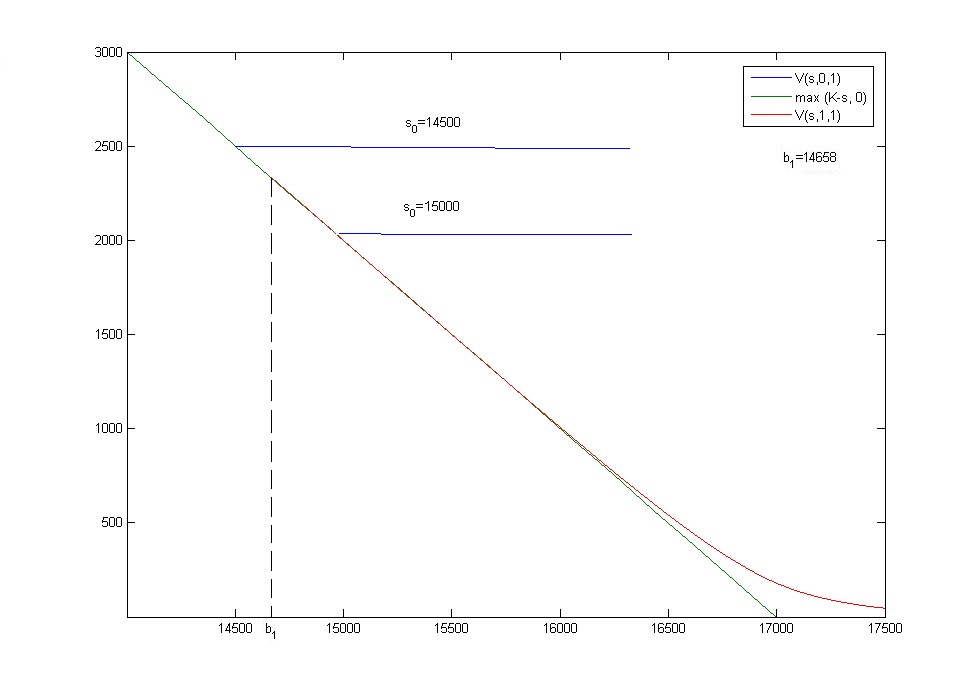}
\end{center}
\vspace{-.7cm}
\centerline{\tiny\hspace*{.5cm} Figure 1 \hspace{7cm} Figure 2}

\medskip
Any point $(\mu_0,s_0)$ left (or above) of the curve $b_0(\mu_0)$
is associated with one of the cases (iii)(b,c), 
and any point right (or below) the curve
is associated with one of the cases (iii)(a),(iv) of Theorem \ref{main}. 

Figure 2 shows the value functions 
corresponding to the two points $(-1,14500)$ and $(-1,15000)$ 
which are both left of the curve $b_0(\mu_0)$ in Figure 1.
The green anti-diagonal line is part of the gain function $(K-\cdot)^+$,
and the red convex curve is the graph of $V(\cdot,1,1)$
which merges onto the gain function at $b_1$.
The blue branches hitting $V(\cdot,1,1)$
at $s_0=14500$ and $s_0=15000$, respectively, 
are the graphs of the corresponding version of $V(\cdot,0,1)$
before the regime change at $s_0$.

Recall that $V(\cdot,0,1)$ and $V(\cdot,1,1)$ are two different components
of the value function, and they only meet continuously in the above picture
because of the boundary condition 
explained in Remark \ref{explain generator}(i) on page \pageref{explain generator}.

Figure 2 can be used to illustrate the qualitative difference 
between the strategy assuming $s_0=14500\le b_1$ (case (iii)(c))
and the one assuming $s_0=15000>b_1$ (case (iii)(b)).
When assuming $s_0=14500$, the trader waits for the index to reach $s_0$
and would then sell/exercise the put immediately. When assuming $s_0=15000$,
they would also wait for the index to reach $s_0$ but would then exploit
the regime change from $\sigma_0$ to $\sigma_1$ implemented into their model
due to an implied leverage effect during a market fall: they would
either sell/exercise the put after a further waiting time of the order of $1/\lambda$,
or sell/exercise the put when the index reaches $b_1$.
Note that, in our example, $1/\lambda$ equals 1/2 week which is very short. 
As a consequence, $V(\cdot,1,1)$ looks very similar to how 
the value function of a traded American put shortly before maturity would look like, 
and this explains why $V(\cdot,1,1)$ is so close to the gain function.
\begin{remark}\rm\label{robust region}
\begin{itemize}\item[(i)]
According to our definition of the value function,
all strategies refer to exercising the option. However,
since the price of an option which has not matured yet always tops its exercise value,
selling the option would not cause any disadvantage.
\item[(ii)]
As argued above, choosing a put with strike $K$ such that
the level $s_0$ defining the trade is below $b_1$ 
and then applying Theorem \ref{main} using a small value of $\mu_0$, that is,
using a value of $\mu_0$ in accordance with an expected market drop
results in an optimal strategy 
where the trader would NOT benefit from the implied leverage effect.
So, the trader would want to choose $K$ such that the level $s_0$
they have in mind is above $b_1$.
\item[(iii)]
Following (ii),
the strategy to be used would be the one described above
with respect to $s_0=15000$ (case (iii)(b)) 
for at least all $(\mu_0,s_0)$ in the marked area shown in Figure 1.
Notice that
this area covers values of $\mu_0$ as large as $13.7\%$ 
which is, for example, 
well-above the 10 years (2004-2013) average return rate of $6.05\%$  
of the Dow Jones.
Thus, the strategy given in case (iii)(b) of Theorem \ref{main}
is \underline{robust} in the sense that it applies to small values of $\mu_0$,
when the model would predict a market drop
in accordance with the preferences of the trader,
but also to `neutral' values of $\mu_0$,
when the model would predict standard returns rather than a market drop.
\end{itemize}
\end{remark}

Because of Remark \ref{robust region}(ii), 
we restrict the remaining part of our discussion to cases where $s_0>b_1$.
By Remark \ref{robust region}(iii), we know that the strategy given
in case (iii)(b) is robust for small and `neutral' values of $\mu_0$.
Next, we discuss the type of strategy offered by Theorem \ref{main}  
when the trader's preferences for the future are `entirely' wrong,
that is, when $\mu_0$ is significantly bigger than $13.7\%$
and $(\mu_0,s_0)$ belongs to the quadrant on the right-hand side
of the marked area in Figure 1.
For demonstration, we choose $\mu_0=30\%$.

Figure 3 shows the part of the quadrant on the right-hand side
of the marked area in Figure 1 which refers to $10\%\le\mu_0\le 100\%$.
The blue upper concave curve is the graph of $b_0(\mu_0)$,
and the green concave curve beneath,
which meets the upper curve at $\mu_0=13.7\%$, 
is the graph of a function we call $s_0^{max}=s_0^{max}(\mu_0)$. 
This function gives the root of the equation
\[
V(s_0,1,1)\,=\,(K-b_0)\bigg(\dfrac{s_0}{b_0}\bigg)^{\gamma^-},
\quad
\mbox{$s_0$ unknown,}
\]
which is the value of $s_0$ at which the switch between 
case (iii)(a) and case (iv) of Theorem \ref{main} occurs.
The red horizontal line again marks the level of $b_1=14658$,
and the black vertical fat bar marks 
the values of $s_0$ between $b_1$ and $s_0^{max}=15742$ at $\mu_0=30\%$.
\begin{center}
\includegraphics[width=7.5cm]{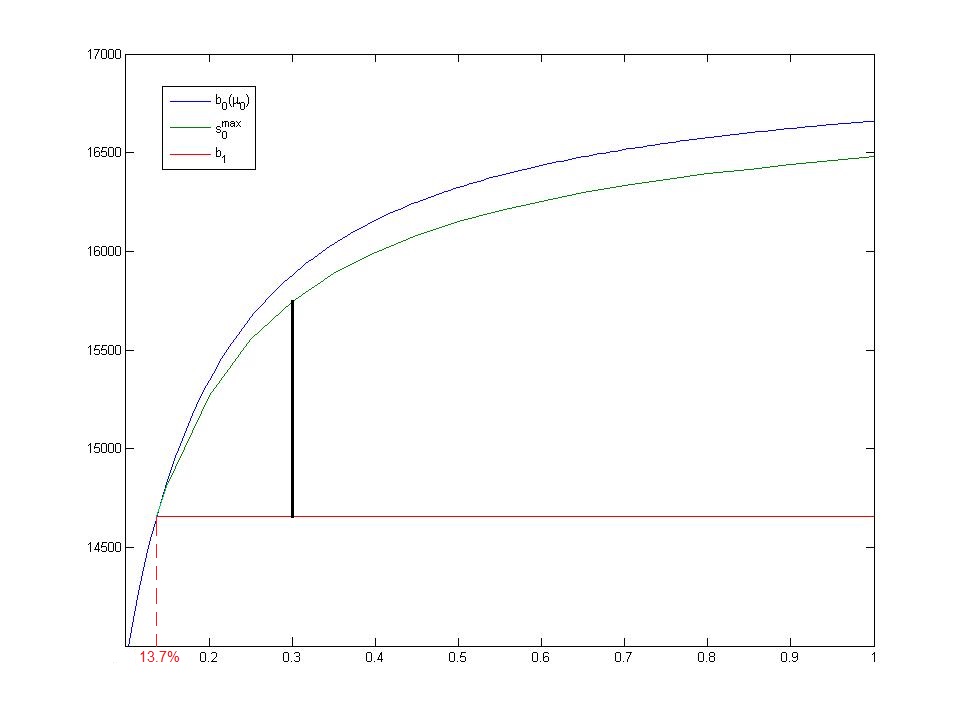}
\hspace{-.5cm}
\includegraphics[width=7.5cm]{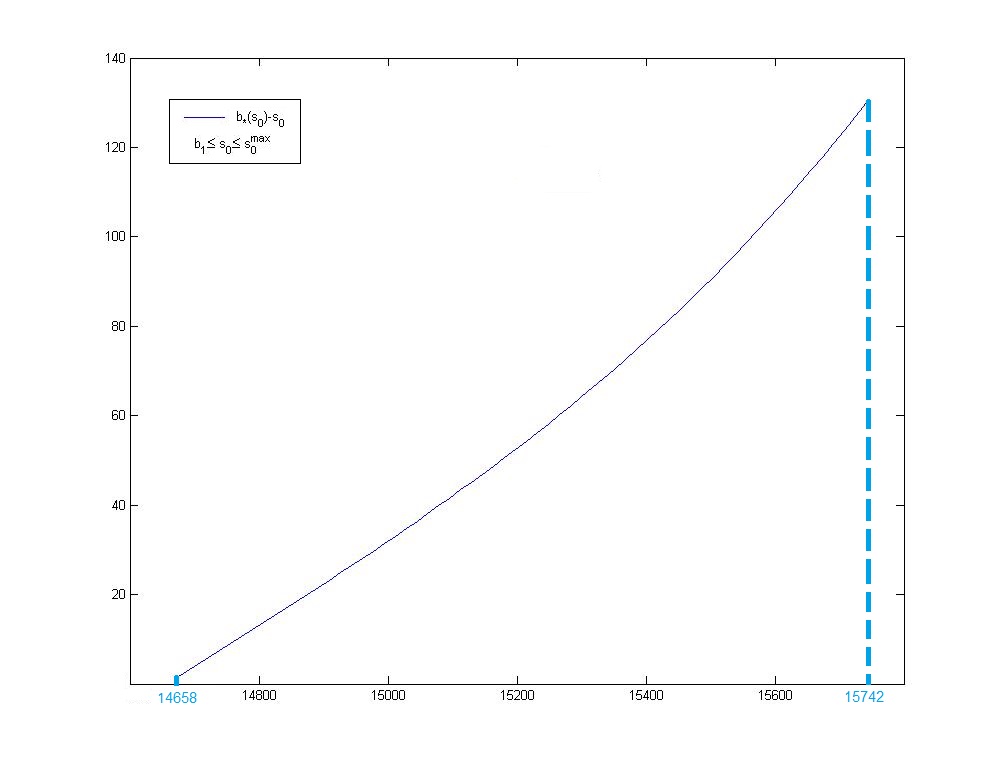}
\end{center}
\vspace{-.7cm}
\centerline{\tiny\hspace*{.5cm} Figure 3 \hspace{7cm} Figure 4}

\medskip
Any point $(\mu_0,s_0)$ between the horizontal line and the curve $s_0^{max}(\mu_0)$
is associated with case (iv) of Theorem \ref{main},
while any point between the two curves $s_0^{max}(\mu_0)$ and $b_0(\mu_0)$
is associated with case (iii)(a).

In case (iv), there exists a corresponding $b_*=b_*(\mu_0,s_0)\in(s_0,b_0)$.
For fixed $\mu_0=30\%$, we write $b_*(s_0)$ for $b_*(0.3\,,s_0)$,
and Figure 4 shows the graph of $b_*(s_0)-s_0$ for those values of $s_0$
marked by the vertical fat bar in Figure 3.
Note that a further look at the proof of Lemma \ref{choice b*} reveals
$\lim_{s_0\uparrow s_0^{max}}b_*(s_0)=b_0$.

Figure 5 below shows the value function corresponding to the point $(0.3\,,15000)$
which is a point on the vertical fat bar in Figure 3. To better illustrate the typical shape
of the components of this value function, we scaled the axes in a non-linear way
which is why, in contrast to the other figures,
there are no numerical values assigned to the axes.
\begin{center}
\includegraphics[width=7.5cm]{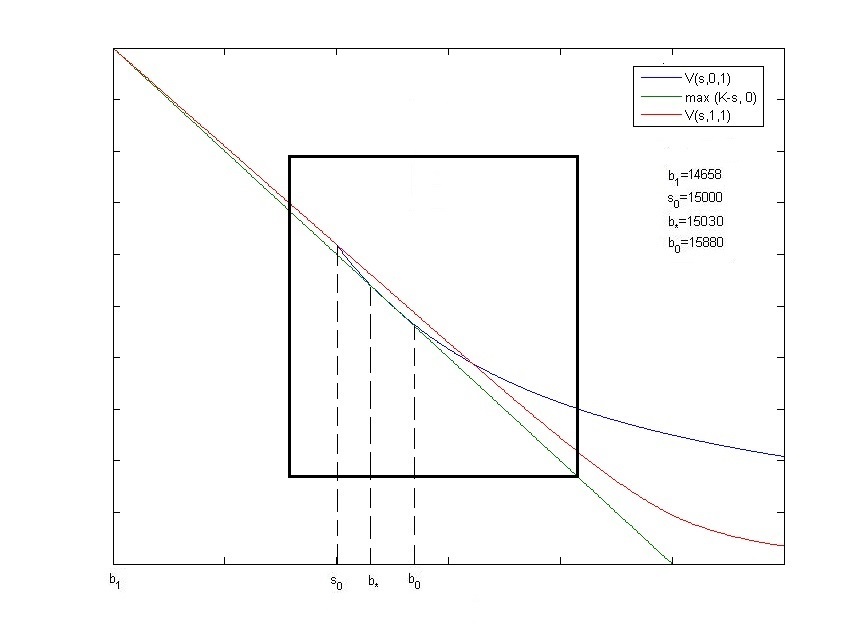}
\hspace{-.5cm}
\includegraphics[width=7.5cm]{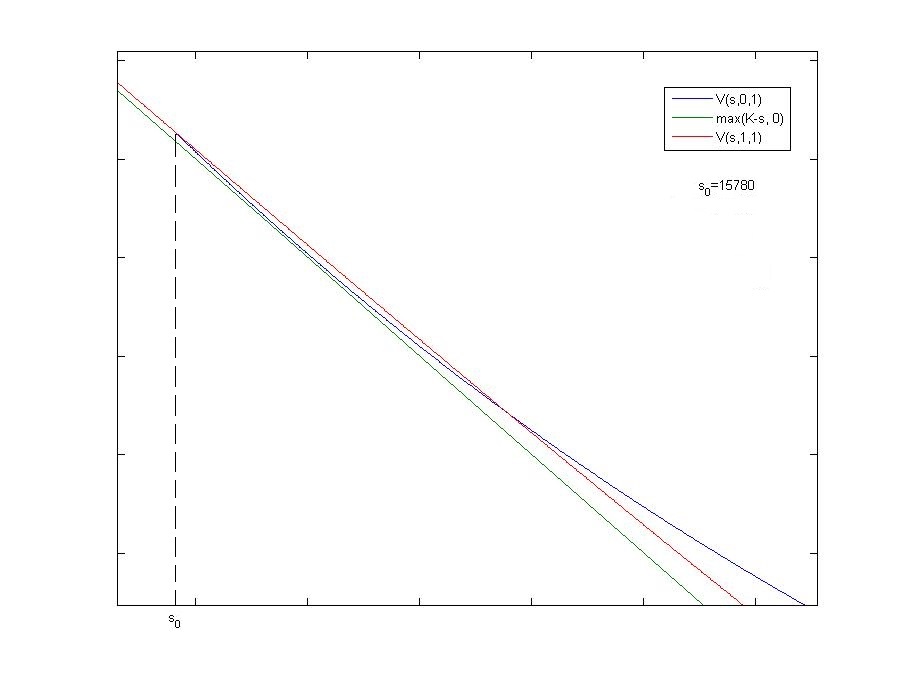}
\end{center}
\vspace{-.7cm}
\centerline{\tiny\hspace*{.5cm} Figure 5 \hspace{7cm} Figure 6}

\medskip
The green anti-diagonal line is part of the gain function $(K-\cdot)^+$,
and the red convex curve is the graph of $V(\cdot,1,1)$
which merges onto the gain function at $b_1$ in the upper left corner.
The component $V(\cdot,0,1)$, plotted in blue, is identical to the gain function 
between $b_*$ and $b_0$ but also has two branches: 
the left branch below $V(\cdot,1,1)$ connects $V(\cdot,1,1)$ at $s_0$ 
with the gain function at $b_*$, 
and the right branch crosses $V(\cdot,1,1)$ 
before merging onto the gain function at $b_0$.

Figure 6 zooms into the window marked in Figure 5
showing the components of the value function corresponding to the point $(0.3\,,15780)$
which is a case-(iii)(a)-point 
above the vertical fat bar in Figure 3 but still below the curve $b_0(\mu_0)$.
Only the component $V(\cdot,0,1)$ changes. While, in Figure 5, 
$V(\cdot,0,1)$ is identical to the gain function on a whole interval $(b_*,b_0)$,
in Figure 6, its graph stays above the gain function everywhere crossing $V(\cdot,1,1)$
from the right and meeting it again further left at $s_0$.

Figure 6 graphically confirms Theorem \ref{main} in asserting that
the case-(iii)(a)-strategy is identical to the case-(iii)(b)-strategy
discussed in the context of Figure 2. All in all, the case-(iii)(b)-strategy
would be applicable for all $(\mu_0,s_0)$ in both
the marked area shown in Figure 1
and in that part of the quadrant on the right-hand side
of this marked area which is above the curve $s_0^{max}(\mu_0)$ in Figure 3.
\begin{remark}\rm\label{don't wait}
Recall that the case-(iii)(b)-strategy involves waiting for the index 
to fall $s-s_0$ points where, by Remark \ref{world}(ii), the size of $s-s_0$ is considerable.
Thus, for values of $\mu_0$ as large as $30\%$ in our example,
one would expect the index to take a rather long time for dropping as much as $s-s_0$.
To avoid this risk,
the trader would not want to choose a put with strike $K$
such that the level $s_0$ defining the trade
is above the curve $s_0^{max}(\mu_0)$ in Figure 3
for a range of `larger' values of $\mu_0$.
\end{remark}

The alternative to this unsuitable choice of $K$ 
would be to choose a put with strike $K$ such that
the level $s_0$ stays below the curve $s_0^{max}(\mu_0)$ in Figure 3
for all `larger' values of $\mu_0$.
This alternative refers to the remaining case-(iv)-strategy,
and we return to Figure 5 to discuss this strategy in more detail.

Recall that $\mu_0=30\%$ and $s_0=15000$ in our example.
According to the function $b_*(s_0)-s_0$ shown in Figure 4,
the gap between $s_0$ and $b_*$ in our example
is about 30 points of the Dow Jones index.
Clearly, if $s-s_0$ is of the order of 50 points
and the value $s$ of the Dow Jones is of the order of 15000 points,
then a drop from $s$ to $s_0$ would not have any effect on the volatility of the index.
So, at least in our example, to be in agreement with the model's assumptions,
the value $s$ should be well above $b_*$ (i.e.\ $b_*\ll s$).

Figure 5 can now be used to illustrate the two different strategies depending 
on how much the present value $s$ is above $b_*$.
When assuming $b_*<s<b_0$, the trader would sell/exercise immediately,
while, when assuming $b_0<s$, the trader waits for the index to reach $b_0$
and would then sell/exercise. By the same reason given in Remark \ref{don't wait},
the trader would not want to wait for the index to reach $b_0$ 
if $\mu_0$ is as large as $30\%$. Therefore, 
a further but final constraint on where the present value $s$
should be located is $b_*\ll s<b_0$.
Note that $b_0-b_*$ is of the order of 800 points in our example
which is on the right scale for taking into account a possible leverage effect
if the index drops from $s$ satisfying $b_*\ll s<b_0$ 
to a level $s_0$ below $b_*=15030$.
For given $s$ and $s_0$, the relation $b_*({\mu}_0,s_0)<s<b_0({\mu_0})$
wanted for all `larger' values of $\mu_0$ can be achieved 
by choosing an appropriate strike level $K$.
\subsection{Choosing the Strike Level}
The following steps present a summary of the previous discussion
on how to choose the strike level of the put
depending on both
the stopping levels $b_0,b_1,b_*$ given by Theorem \ref{main}
and the level $s_0^{max}$ introduced in the paragraph preceding Figure 3 above.
After this summary we briefly describe how the strategies given in
Theorem \ref{main}
could be used for trading.
\begin{itemize}
\item[{\bf Step 1:}]
Fix a discount rate $\alpha$ and choose an index with present value $s$.
Find $\sigma_0$ by comparing implied volatilities calculated from
a range of traded options on the index. Decide about the size of $s-s_0$ 
the index is expected to drop in the near future.
Based on analysing historical data or otherwise, decide about the size 
of the `excited' volatility $\sigma_1$. Analysing historical data or otherwise,
find the average time span of an `excited' volatility regime
after a drop of size $s-s_0$ of the index of your choice, that is,
find $1/\lambda$. Set $\mu_1=0$.
\item[{\bf Step 2:}]
For different values of $K$ calculate:
$b_1$; $\tilde{\mu}_0$ such that $b_1=b_0(\tilde{\mu}_0)$;
$s_0^{max}(\tilde{\mu}_0+\rho_0)$ for sufficiently large $\rho_0$;
$b_0(\tilde{\mu}_0+\rho_0)$; $b_*(\tilde{\mu}_0+\rho_0,s_0)$.
For the right tuning of $\rho_0$ compare with both Figure 3 where 
$\tilde{\mu}_0$ and $\tilde{\mu}_0+\rho_0$ were $13.7\%$ and $30\%$, respectively,
and the comments in Remark \ref{robust region}(iii) about the magnitude of $13.7\%$.
\item[{\bf Step 3:}]
Finally choose a put with strike level $K$ such that 
$b_1<s_0<s_0^{max}(\tilde{\mu}_0+\rho_0)$
and
$b_*(\tilde{\mu}_0+\rho_0,s_0)<s<b_0(\tilde{\mu}_0+\rho_0)$.
We think that, for trading, the present value $s$ of the index and the drop-to-level $s_0$
would be placed best if the size of $b_*(\tilde{\mu}_0+\rho_0,s_0)-s_0$
is on a smaller scale than $s-s_0$ as in our example above.
Note that this would also entail $b_*(\tilde{\mu}_0+\rho_0,s_0)\ll s$.
\end{itemize}

When trading a put of the above choice using the strategies given by
Theorem \ref{main},
assuming that the value of $\mu_0$ is sufficiently negative
to be conform with a drop of size $s-s_0$ in the near future,
the trader would initially follow the case (iii)(b) strategy.

First, if the level $s_0$ is reached within the expected time frame,
the trader would continue following the case (iii)(b) strategy to the
end.
In practical terms, the exponential waiting time should be realised
by waiting a multiple of the average waiting time $1/\lambda$
where the choice of the multiple is up to the trader.

Second, if the level $s_0$ is \underline{not} reached within the
expected time frame,
the trader would have gained enough new market data to update
the value of $\mu_0$. Based on statistical testing or otherwise,
they should decide whether the updated value of $\mu_0$ is below
$\tilde{\mu}_0$
or above $\tilde{\mu}_0+\rho_0$.

If the decision is for the updated $\mu_0$ to be below
$\tilde{\mu}_0$,
the trader could continue following the case (iii)(b) strategy
(updating $\mu_0$ again if necessary),
but they should also consider to finish the trade as soon as
selling/exercising
would not result in any losses.

If the decision is for the updated $\mu_0$ to be above
$\tilde{\mu}_0+\rho_0$,
the trader should change to the case (iv) strategy but with respect to
the most recent value of the underlying asset. If this \underline{new}
present value $s$
is above $b_*(\tilde{\mu}_0+\rho_0,s_0)$, they should sell/exercise
immediately.
However, if it is in the range of $s_0$ to
$b_*(\tilde{\mu}_0+\rho_0,s_0)$,
the trader could continue following the primary case (iii)(b)
strategy,
unless the level $b_*(\tilde{\mu}_0+\rho_0,s_0)$ is reached before the
drop-to-level $s_0$
when they should sell/exercise immediately.
\section*{Appendix}
We verify that the explicit expression given for $V(\cdot,1,1)$ 
in Theorem \ref{main}(ii)
is indeed the value function. Our method of verification is going to be different to the
standard method mentioned in Remark \ref{existingLit}(iii).

First, we introduce the formal differential operator
\[
L_1 f(s)\,=\,\mu_1 s f'(s)+\frac{1}{2}\sigma_1^2 s^2 f''(s)
+\lambda\left((K-s)^+ - f(s)\right)
\]
and remark that, by standard arguments
(see Section 5.2.1 in \cite{pham} for example),
the function $V(\cdot,1,1)$ is continuous and satisfies the variational inequality
\[
\min\left((\alpha I-L_1)V(\cdot,1,1)\,,\,V(\cdot,1,1)-(K-\cdot)^+\right)\,=\,0
\quad\mbox{on}\quad
(0,\infty)
\]
in viscosity sense.

Second, 
given on a family of probability spaces
$(\tilde{\Omega},\tilde{\mathcal{F}},\tilde{\pp}_{\!s},s>0)$,
we consider the Feller process $(\tilde{S}_t,\,t\ge 0)$
whose generator is the closure of
\[
\tilde{L}_1f=\mu_1s f'+\dfrac{1}{2}\sigma_1^2s^2f'',
\quad
f\in C_0^2((0,\infty)),
\]
and define the value function
\[
\tilde{V}_1(s)\,=\,\sup_{\tilde{\tau}\ge 0}
\tilde{\ee}_s[e^{-\tilde{\beta}\tilde{\tau}}(K-\tilde{S}_{\tilde\tau})^+
\,+\,\lambda\!\int_0^{\tilde{\tau}}e^{-\tilde{\beta}u}(K-\tilde{S}_u)^+\,\dd u],
\]
where $\tilde{\beta}=\alpha+\lambda$ and the supremum is taken over all
stopping times with respect to $(\tilde{S}_t,\,t\ge 0)$.

Third, by \cite[Thm.5.2.1]{pham}, the value function $\tilde{V}_1$
is the unique viscosity solution to the variational inequality 
\[
\min\left((\tilde{\beta}I-\tilde{L}_1)f-\lambda(K-\cdot)^+\,,\,f-(K-\cdot)^+\right)\,=\,0
\quad\mbox{on}\quad
(0,\infty)
\]
satisfying a linear growth condition.
Therefore, since 
\[
\alpha I-L_1\,=\,\tilde{\beta}I-\tilde{L}_1 -\lambda(K-\cdot)^+,
\]
and since $V(\cdot,1,1)$ is bounded, it follows from the above first step that $V(\cdot,1,1)=\tilde{V}_1$.

Fourth, by \cite[Lemma 5.2.2]{pham}
(uniform ellipticity on compact subsets of $(0,\infty)$ is sufficient in our case),
the value function $\tilde{V}_1$ is $C^2$  inside the continuation region,
and, by \cite[Prop.5.2.1]{pham}, it is $C^1$ on the boundary.
As a consequence, the value function $V(\cdot,1,1)=\tilde{V}_1$ is $C^1$ 
on the whole domain $(0,\infty)$ and $C^2$ inside of
$\{s>0:V(\cdot,1,1)>(K-s)^+\}$.

If one can now prove that the set $\{s>0:V(\cdot,1,1)>(K-s)^+\}$ must have the form $(b_1,\infty)$
for some optimal stopping level $b_1\in(0,K)$,
then $V(\cdot,1,1)$ would satisfy both the boundary and pasting conditions (\ref{5})
and, in classical sense, the equations above (\ref{5}) on page \pageref{5}.
But the latter must have the solution given by (\ref{classical solution})
leading to the explicit expression given for $V(\cdot,1,1)$ in Theorem \ref{main}(ii).

\medskip
\noindent
{\bf Remark A.1.}\label{not unique}
The optimal stopping level $b_1$ must satisfy (\ref{99})
because $V(\cdot,1,1)$ satisfies (\ref{5}).
However, finding $b_1$ by solving (\ref{99}) requires showing uniqueness of solutions
to a non-linear equation.
This uniqueness problem was neither addressed in \cite{Guo} nor in \cite{Rogers}.
For completeness, we are going to show uniqueness of solutions to (\ref{99})
in Lemma {\bf A.3} after the next lemma.
However, Example {\bf A.4} at the end of the Appendix shows that this
uniqueness is \underline{not} an intrinsic property of equations like (\ref{99})
even if the equations were derived from an optimal stopping problem
with convex gain function.

\medskip
On the whole, finishing our method of verification, we only need to prove
the following lemma.

\medskip
\noindent
{\bf Lemma A.2.}
{\it There exists $b_1\in(0,K)$ such that}
\[
\{s>0:V(\cdot,1,1)>(K-s)^+\}\,=\,(b_1,\infty).
\]
\begin{proof}
The proof can be divided into the two cases $\mu_1\le\alpha$ and $\mu_1>\alpha$.

For $\mu_1\le\alpha$, the process $(e^{-\alpha t}\,S_t,\,t\ge 0)$
is a $\pp_{\!\!s,1,1}$\,-\,supermartingale, and the conclusion of the lemma
follows by copying the proof of Proposition 1 in \cite{Rogers}.

For $\mu_1>\alpha$, we present a proof which works for $\mu_1\ge 0$.

First, since our gain function is $(K-\cdot)^+$, one has, for all $\mu_1$, that
$[K,\infty)\times\{1\}\times\{1\}$ is in the continuation region.
As a consequence,
\[
b_1\,\stackrel{\mbox{\tiny def}}{=}\,
\sup\{s>0:V(s,1,1)=K-s\}\,<\,K
\]
because the stopping region is closed when both value function and gain function are continuous.

We now want to prove by contradiction that $b_1>0$ and that $V(s,1,1)=K-s$,
for all $s\in(0,b_1]$, finishing the proof of the lemma.

Recall the regularity properties of $V(\cdot,1,1)$ discussed in the first part
of the Appendix. So, if $b_1=0$, that is, if the supremum is taken over the empty set, then
\[
V(s,1,1)\,=\,
c_1 s^{\beta^+}+\;c_2 s^{\beta^-}
+\;h(s),
\quad\mbox{for}\;s\in(0,K),
\]
by the same arguments which led to (\ref{classical solution}) in Section \ref{guo case}.
Since $V(\cdot,1,1)$ is bounded, the coefficient $c_2$ would have to be zero, 
and hence, using the explicit choices for $h$ given in Remark \ref{existingLit}(ii),
\[
\lim_{s\downarrow 0}V(s,1,1)\,=\,\frac{\lambda K}{\alpha+\lambda}\,<\,K
\]
which would contradict $V(\cdot\,,1,1)\ge(K-\cdot)^+$ on $(0,\infty)$.

Next, assume $b_1>0$ and $V(s,1,1)>K-s$ for some $s\in(0,b_1]$.

Then there would exist a component $(u_1,u_2)$ of the continuation region,
where\footnote{One can rule out $u_1=0$ the same way $b_1=0$ was ruled out.}
$0<u_1<u_2<b_1$, such that
\[
\mu_1 s\partial_1V(s,1,1)
+\frac{1}{2}\sigma_1^2 s^2\partial_{11}V(s,1,1)
+\lambda(K-s)
-(\alpha+ \lambda )V(s,1,1)
\,=\,0,\eqno({\rm A.1})
\]
for all $s\in(u_1,u_2)$, in classical sense, and
\[
V(u_1,1,1)=K-u_1,\;
\partial_1 V(u_1,1,1)=-1,\;
V(u_2,1,1)=K-u_2,\;
\partial_1 V(u_2,1,1)=-1.\eqno({\rm A.2})
\]

Note that, since $V(\cdot,1,1)$ dominates $(K-\cdot)^+$, $({\rm A.1})$ implies
\[
\partial_{11}V(s,1,1)\,\ge\,
-\frac{2\mu_1}{s\sigma_1^2}\,\partial_1V(s,1,1)
+\frac{2\alpha}{s^2\sigma_1^2}\,V(s,1,1),
\quad\mbox{for $s\in(u_1,u_2)$,}
\]
leading to
\[
\partial_{11}V(s,1,1)\,\ge\,
\frac{2\alpha}{u_2^2\,\sigma_1^2}\,V(s,1,1)
\,\ge\,
\frac{2\alpha}{u_2^2\,\sigma_1^2}\,(K-u_2)\,>\,0,
\quad\mbox{for $s\in U_2\cap(u_1,u_2)$},
\eqno({\rm A.3})
\]
because $\mu_1\ge 0$ and, by continuity, 
$\partial_1 V(\cdot,1,1)$ is negative in a neighbourhood $U_2$ of $u_2$.

Thus, $\partial_1 V(\cdot,1,1)$ is strictly increasing in a left neighbourhood of $u_2$
so that, for $({\rm A.2})$ to be true,
there must exist a largest inflection point $u_0\in(u_1,u_2)$ given by
\[
u_0\,=\,\sup\{s\in(u_1,u_2):\partial_{11}V(s,1,1)=0\}.
\]
Of course, by continuity of $\partial_{11}V(\cdot,1,1)$ in the continuity region,
$\partial_{11}V(u_0,1,1)$ equals zero.
But, $\partial_{1}V(u_0,1,1)<-1$
since $\partial_1 V(\cdot,1,1)$ is strictly increasing on $(u_0,u_2)$,
and hence, by the same arguments leading to $({\rm A.3})$,
\[
\partial_{11}V(u_0,1,1)
\,\ge\,
\frac{2\alpha}{u_2^2\,\sigma_1^2}\,(K-u_2)\,>\,0
\]
which contradicts $\partial_{11}V(u_0,1,1)=0$.
\end{proof}

\noindent
{\bf Lemma A.3.}
{\it There is exactly one solution
$(c_1,c_2,d_2,b_1)\in\bfR^3\times(0,K)$
to the system (\ref{99}).}
\begin{proof}
We only show the lemma in the case of $\alpha+\lambda\not=\mu_1$
using the corresponding function $h$ given in Remark \ref{existingLit}(ii)
because nothing fundamental changes when doing the calculation 
in the remaining \underline{single} case of $\alpha+\lambda=\mu_1$
with another function $h$.

First, we ignore the last equation of the system (\ref{99}) 
and replace $b_1$ by an arbitrary $b>0$. 
The resulting system of equations reads 
\begin{align*}
c_1K^{\beta^+}+\;c_2K^{\beta^-}
-\;\dfrac{\lambda K}{\alpha+\lambda-\mu_1}
+\dfrac{\lambda K}{\alpha+\lambda}
&\,=\,d_2K^{\beta^-},\\
c_1\beta^+K^{\beta^+}+\;c_2\beta^-K^{\beta^-}-\;\dfrac{\lambda K }{\alpha+\lambda-\mu_1}
&\,=\,d_2\beta^-K^{\beta^{-}},\\
K-b&\,=\,c_1b^{\beta^+}+\;c_2b^{\beta^-}
-\;\dfrac{\lambda b}{\alpha+\lambda-\mu_1}
+\frac{\lambda K}{\alpha+\lambda}\,,
\end{align*}
and, for each $b>0$, 
this system admits a unique solution $c_1,c_2,d_2$.

To analyse the non-linear last equation of the system (\ref{99}),
we only need to know $c_1$ and $c_2$ which are explicitly given by
\begin{align*}
c_1&\,=\,\frac{\lambda K^{1-\beta^+}}{(\beta^--\beta^+)}
\bigg[\frac{\beta^-}{\alpha+\lambda-\mu_1}-\frac{\beta^-}{\alpha+\lambda}
-\frac{1}{\alpha+\lambda-\mu_1}\bigg],\\
\rule{0pt}{20pt}
c_2(b)&\,=\,\bigg[K-b+\frac{\lambda b}{\alpha+\lambda-\mu_1}
-\frac{\lambda K}{\alpha+\lambda}-c_1b^{\beta^+}\bigg]\frac{1}{b^{\beta^-}}\,.
\end{align*}

Second, for each $b>0$, we introduce the function
\begin{equation*}
V_b(s)\,=\,
c_1s^{\beta^+}+\;c_2(b)s^{\beta^-}-\;
\dfrac{\lambda s}{\alpha+\lambda-\mu_1}+
\dfrac{\lambda K}{\alpha+\lambda}\,,
\quad s>0,
\end{equation*}
and remark that $b>0$ satisfies
\begin{equation*}
-b=c_1\beta^{+}b^{\beta^+}+\;c_2(b)\beta^{-}b^{\beta^-}
-\;\dfrac{\lambda b}{\alpha+\lambda-\mu_1}
\end{equation*}
if and only if 
\[
\frac{\dd}{\dd s}\,V_b(s)|_{s=b}
\,=\,V'_b(b)\,=\,-1.
\]
Therefore, when setting $\Gamma(b)=V'_b(b)$ for $b>0$,
the proof of the lemma reduces to showing that
the equation $\Gamma(b)=-1$ has exactly one root between zero and $K$,
and this will be shown next.

By straight forward calculation, we have that
\begin{equation*}
\Gamma(b)\,=\,c_1b^{\beta^+-1}(\beta^+-\beta^-)
+\frac{\alpha K\beta^-}{(\alpha+\lambda)b}
-\frac{(\alpha-\mu_1)\beta^- +\lambda}{\alpha+\lambda-\mu_1},
\quad\mbox{for}\;b>0,
\end{equation*}
so that
\begin{equation*}
\lim\limits_{b\rightarrow 0}\Gamma(b)=-\infty\;\mbox{(since $\beta^-<0$)}
\quad\mbox{and}\quad
\Gamma(K)=0.
\end{equation*}
Thus, by Intermediate Value Theorem, 
there exists $b_1\in(0,K)$ such that $\Gamma(b_1)=-1$. 

For uniqueness, one only has to show that $\Gamma(\cdot)$ is increasing
on $(0,K)$, that is, $\Gamma'(b)\not=0$ for all $b\in(0,K)$.

But,
\begin{equation*}
\Gamma'(b)\,=\,
c_1(\beta^+-1)(\beta^+-\beta^-)b^{\beta^+-2}
-\frac{\alpha K\beta^-}{\alpha+\lambda}\,b^{-2},
\end{equation*}
and hence, for $b\in(0,K)$,
the equality $\Gamma'(b)=0$  is equivalent to
\[
\left(\frac{b}{K}\right)^{\beta^+}
=\;
\frac{\alpha(\alpha+\lambda-\mu_1)}
{\frac{1}{2}\sigma_1^2\lambda(\beta^+-1)(\beta^- -1)}\,,
\]
where we have used that $\beta^-$ is a root of the equation (\ref{100}).
As the above right-hand side is always negative,
because $\beta^+>1$ if $\alpha+\lambda>\mu_1$ and
$\beta^+<1$ if $\alpha+\lambda<\mu_1$,
there is no $b>0$ such that $\Gamma'(b)=0$.
\end{proof}

\medskip
\noindent
{\bf Example A.4.}
Let $g:(0,\infty)\to\bfR$ be a bounded smooth convex function
satisfying 
\[
g(x)\,=\left\{\begin{array}{rcl}
4&:&x=1/2\\1&:&x=1\end{array}\right.,
\quad
g'(x)\,=\left\{\begin{array}{rcl}
-8&:&x=1/2\\-1&:&x=1\end{array}\right.,
\quad
g(x)=0,\,x\ge 3,
\]
and consider the value function
$
V(x)=\sup_{\tau\ge 0}
\ee[e^{-\tau}\,g(X_\tau^x)]
$,
where $X^x_t=x e^{\sqrt{2}B_t},\,t\ge 0$, is a geometric Brownian motion
on a probability space $(\Omega,{\cal F},\pp)$.

Following the method used in \cite{Guo,Rogers},
the free-boundary value problem associated with this optimal stopping problem is
\[
0\,=\,xV'(x)+x^2V''(x)-V(x),
\quad\mbox{for}\;x>x_0,
\]
subject to
\[
V(x_0)=g(x_0),\quad V'(x_0)=g'(x_0),\quad \lim_{x\to\infty}V(x)=0.
\]

As any solution to this problem must have the form $c_1x+c_2x^{-1}$,
the above boundary and pasting conditions result in $c_1=0$ and two equations
\[
g(x_0)\,=\,c_2x_0^{-1},\quad g'(x_0)\,=\,-c_2x_0^{-2}
\]
for the pair of unknowns $(c_2,x_0)$.

There are at least two solutions to these equations,
$(c_2,x_0)=(1,1)$ and $(c_2,x_0)=(2,1/2)$,
but there might be even more.
Note that the value function is unique and can only be identical to one
of the candidate value functions build from these solutions $(c_2,x_0)$.
Hence this example indeed justifies Remark {\bf A.1} on page \pageref{not unique}.

\end{document}